\documentclass[runningheads]{llncs}

\usepackage[T1]{fontenc}
\usepackage{mathtools}

\usepackage{amssymb}

\input bksymbol2.sty

\newcommand{\N}{\mathbb{N}}

\newcommand{\E}{\mathbf{E}}
\newcommand{\Rb}{\mathbb{R}}

\newcommand{\T}{\mathcal{T}}
\newcommand{\Rf}{\mathcal{R}}
\newcommand{\cak}{\mathfrak{K}}

\newcommand{\Argmax}[1]{\underset{#1}{\textrm{argmax}}}
\newcommand{\Max}[1]{\underset{#1}{\textrm{max}}\  } 
\newcommand{\Sup}[1]{\underset{#1}{\textrm{sup}}}
\newcommand{\Inf}[1]{\underset{#1}{\textrm{inf}}}
\newcommand{\Min}[1]{\underset{#1}{\textrm{min}}\ }
\newcommand{\DX}[1]{\ \mathbf{d} #1}

\newcommand{\Rred}[1]{\textcolor{red}{#1}}

\newcommand{\ALT}[1]{}
\newcommand{\FNS}[1]{}

\newcommand{\MSPP}{{\tt SPC}}
\newcommand{\SAMP}{{\tt SAMP}}

\newcommand{\sd}{{\Gamma}}

\renewcommand{\ee}{\varepsilon}

\newcommand{\mujw}{\mu_{j,w}}
\newcommand{\mujv}{\mu_{j,v}}
\newcommand{\Djw}{\caD_{j,w}}
\newcommand{\Djv}{\caD_{j,v}}
\newcommand{\omuj}{\ovl{\mu_j}}
\newcommand{\omu}{\ovl{\mu}}

\newcommand{\hmu}{\hat{\mu}}
\newcommand{\hcaD}{\hat{\caD}}
\newcommand{\supp}{{\tt supp}}
\newcommand{\htau}{\hat{\tau}}

\newcommand{\KURZ}[2]{#1}

\begin{document}

\title{Improving Statistical Privacy by Subsampling\thanks{
   This research has been conducted within the AnoMed project (https://anomed.de/)
    funded by the BMBF  (German Bundesministerium für Bildung und Forschung)
   and the European Union in the NextGenerationEU action.}}

\titlerunning{Improving Statistical Privacy by Subsampling}

\author{Dennis Breutigam \and R\"udiger Reischuk}

\authorrunning{D. Breutigam, R\"udiger Reischuk}
%
\institute{Institut f\"ur Theoretische Informatik, Universit\"at zu L\"ubeck,
L\"ubeck, Germany \\
\email{\{d.breutigam,ruediger.reischuk\}@uni-luebeck.de}}

\maketitle

\begin{abstract}
Differential privacy (DP) considers a scenario, where an adversary 
has almost complete information about the entries of a database
This worst-case assumption is likely to overestimate the privacy thread for an individual in real life.
Statistical privacy (SP) denotes a setting where only the distribution of the database entries is known
to an adversary, but not their exact values.
In this case one has to analyze the  interaction between  noiseless privacy based on the entropy of 
distributions and privacy mechanisms that distort the answers of queries, which
can be quite complex.

A privacy mechanism often used is to take samples of the data for answering a query.
This paper proves precise bounds how much different methods of sampling 
increase privacy in the statistical setting
with respect to database size and sampling rate.
They allow us to deduce when and how much sampling provides an improvement
and how far this depends on the privacy parameter $\ee$. To perform these
investigations we develop a framework to model sampling techniques.

For the DP setting tradeoff functions have been proposed as a finer measure for privacy 
compared to $(\ee,\dd)$-pairs.
We apply these tools to statistical privacy with subsampling to get a
comparable characterization
\keywords{privacy \and sampling \and tradeoff function}
\end{abstract}



\section{Introduction}
Many machine learning algorithms, such as stochastic optimization methods and
Bayesian inference algorithms include sampling operations. Due to the increasing
demand for such learning algorithms, especially for sensitive data, there have
been recent efforts to investigate the influence of sampling methods with respect to
privacy \cite{BBG18,WBK19,BBG20,IC21,S22}.
However, the standard model \emph{differential privacy} makes
extremely strong assumptions about the power of an adversary
that from answers of queries to a database tries to get information of single entries.
To guarantee privacy in this case a strong distortion of the answers is necessary --
thus the quality deteriorates siginificantly
  \cite{JD22,D21}.
For many real life applications one should consider less powerful adversaries
in the hope to improve the utility of the data.
The prior information, also called background knowledge is limited.
Such settings has been propsed and named 
\emph{noiseless privacy, (inference based) distributional differential privacy} 
and \emph{passive partial knowledge differential privacy} \cite{BBG11,BGKS13,DMK20}.
Such a model requires a more complex mathematical analysis than a worst case scenario.

This paper consider a situation, where only the distribution of the database 
entries is publicly known, but not their exact values nor any further information,
which we call \emph{Statistical privacy} \cite{BR25}.
To consider different types of sampling methods are considered for increasing
statistical privacy we develop a general framework using \emph{sampling
templates}. Combining this with Markov process theory we analyze how the distributions
generated effect privacy and define \emph{statistical privacy curves}. Then for
the analysis we use tools developed by Balle et. al. in \cite{BBG18} and Markov
process theory to bound the $\alpha$-divergence. This establishes precise bounds for the
amplification of privacy by sampling techniques depending on database size and
sampling rate. In addition to the analytical estimation we plots are provided for
different parameter settings that illustrate the improvement, in particular with
respect to the privacy parameter $\ee$. 

For the DP setting tradeoff functions and the notion of $f$-differential privacy
have been proposed by Dong et. al in \cite{DRS22} as a finer measure for
privacy. We apply these tools to statistical privacy with subsampling to get a
comparable characterization.

Our main results are the following: \vspace*{-1.5ex}
\begin{itemize}
    \item a general framework for sampling techniques
    \item a simple analytical formular for the amplification by sampling without
    replacement
    \item an analysis for Poission sampling in the statistical setting that handles
    different sample sizes
    \item managing dependencies appearing in case of sampling with replacement
\end{itemize}

This paper is organized as follows.
In the next section formal definitions.
%
In section \ref{section:subsampling_techniques} we develop the framework for
sampling mechanisms and introduce sampled privacy curves.
%
%
The bounds for the different sampling techniques are proven in 
section \ref{section:subsampling_results}. In the next section tradeoff
functions are transferred to statistical privacy with subsampling. The paper
concludes with.... 

An appendix provides further details and an introduction to the basics of stochastic process
theory used in our proofs.

\section{Preliminaries}\label{section:statistical_privacy}
In this paper we will use the following notions.

\begin{definition}
A database  \wfinw{$D$} of size $n$  is a sequence of $n$ independent entries $  I_1,\ldots,I_n$
where $I_j$ is taken from a set $W$ of possible values for entries.
$I_j$ is distributed according to a distribution $\omuj$ on $W$.
These marginal distributions may be different, but their support has to be identical --
thus we assume that it is equal to $W$.
They imply a distribution, resp.~density function  \wfinw{$ \mu$} on $W^n$.
This is modeled by a random variable $\caD$ on $W^n$ that is distributed according to $\mu$ ($D \sim \mu$).

By $\mujw$ and $\Djw$ we denote the conditional distribution, resp.~random variable where $I_j$ is fixed
to a constant value $w \in W$.



Let $\caF$ be a set of queries that may be asked for a given database.
Formally, this is described by measurable functions \wfinw{$   F: \; W^n \mapsto A$},
where $A$ denotes an appropriate  set of possible answers.
We assume that $W$ and $A$ are totally ordered and call a query monotone if
$D \le \hat{D}$ implies $F(D) \le F(\hat{D})$ where the partial order on $W^n$ is defined by
$I_j \le \hat{I}_j$ for all $j$.

A privacy technique $\caM$ (also called mechanism) 
is a function that maps database queries to random variables on $A$,
that means for a given database $D$ and query $F$, $\caM(F,D)$ 
is a random variable on $A$ that disturbs the correct answer $F(D)$.
If the database itself is a random variable $\caD$ then $\caM(F,\caD)$ 
contains two types of randomness, 
the internal uncertainty of $\caD$ and the external noise imposed by $\caM$.
\wkast
\end{definition}


To measure how much privacy of  entries is lost by a query $F$ one considers
a critical entry $j$  and estimates how much information about its value $I_j$ 
an adversary can deduce from the answer.
The most challenging task would be to reconstruct $I_j$ or part of it.
Privacy investigations typically consider a much simpler decision problem, namely
to differentiate between two neighboring databases, namely whether  
 the entry  $I_j$  equals $v$ or $w$  for  values in $v,w \in  W$
 while keeping all other entries unchanged.
 Alternatively, instead of changing the $j$-th entry one could remove it
 to get a different notion of neighborhood.
 But this does not make much difference. 
When studying sampling techniques it is more convenient not to change
the database size, thus we will only consider replacement of the
$j$-th entry by a different value.


\begin{definition}{\rm\wfinw{Differential Privacy (DP)}\cite{Dwo06}\\
For  a collection   $\caW$ of databases a
privacy technique $M$ achieves $(\ee,\dd)$-differential privacy for a query $F$
if for all neighboring databases $D,D'$ in $\caW$ and  for all $S \subseteq A$ 
it holds
\[  \Prob{M(F,D) \in S}{} \ \kla  e^\ee \ \Prob{M(F,D') \in S}{}   + \delta  \ .   \]
}\end{definition}

To protect privacy one should not allow queries that differentiate between entries.
Furthermore, sampling only makes sense if the random order in which elements are drawn
does not matter.
If one requires that queries $F$ are symmetric this  holds.
But in the distributional setting one can be more general.
We consider functions of the form
\[    F(x_1,\ldots,x_n)  \gla h(g_1(x_1), \ldots, g_n(x_n)) \ ,  \]
where $h: \; W' \to A$ for some set $W'$ is symmetric, but the $g_j: \; W \to W'$ can be arbitrary.
The distribution $\mu$ for vectors $(x_1,\dots,x_n) \in W^n$ induces a distribution $\mu'$ for 
$(y_1,\ldots,y_n) \in W'^{\;n}$ where $y_i=g_i(x_i)$.
Thus, such a query $F$ on $W^n$ with distribution $\mu$ is equivalent to a symmetric query
$F'=h(y_1,\ldots,y_n)$ on $W'^{\;n}$ with distribution $\mu'$. 

Taking samples from a sequence may generate arbitrary orderings of elements. 
Reordering a sample should have no influence. 
Therefore,  a distance measure for database distributions and samples should be invariant against permutations.


\begin{definition}{\rm \wfinw{Distributional Distance}\label{DD}\\
For distributions $ \mu = \omu_1 \times \ldots \times \omu_n$ and
 $\nu = \ovl{\nu}_1 \times \ldots \times \ovl{\nu}_n$ of databases 
their distance   is defined as
\[      \sd(\mu, \nu) \dea
  \min_{\pi \in \Pi_n} \vert \{j \mid \omu_j \neq \ovl{\nu}_{\pi(j)}  \} \vert \ ,     \]
where $\Pi_n$ denotes the set of all permutations on $n$ elements. 

 If the sizes $n,n'$ of the databases differ the distance $\sd(\mu, \nu)$ is computed by
 minimizing over all injective functions $\pi$ from the smaller index set to the larger one
 and adding $n-n'$ to the value obtained.
}\end{definition}

\begin{definition} $ $ \\ {\rm
Let  $\mu,\hmu$  be distributions and $F$ a query.
The the corresponding \wfinw{Privacy Loss Random Variable} (PLRV) \cite{DN03}
   \FNS{ \Rred{$\caD \sim \mu, \caD' \sim \hmu$}} 
 is a real-valued function defined for  $a \in A$ as
    \[ \caL^{F}_{\mu,\hmu}(a) \dea \ln \ \frac{\mu \caK_F (a)}{\hmu \caK_F (a)} \  ,\]
where $\caK_F$ denotes the Markov kernel corresponding to $F$
and $\ln \frac{0}{0} := 0$ and $\ln \frac{>0}{0} := \infty$.
This operator maps a distribution $\mu$ on $W^n$ to the distribution on $A$ that is generated by $F$.

Using the PLRV, a distance measure for $\mu \caK_F$ and $\hmu \caK_F$ 
called \bfinw{privacy curve} \cite{BBG18} has been defined as
\begin{eqnarray*}
\DD_{F,\mu, \hmu}(\ee) 
  &:=&  \int_A \mu   \caK_F (a) \cdot \max \lk \; 0, \; 1-\exp(\ee-\caL^{F}_{\mu,\hmu}(a)) \rk  \DX{a} \\
  &=&  \int_A \max   \lk 0, \; \mu   \caK_F (a) - e^{\ee} \; \hmu \caK_F (a) \rk \DX{a}     \\           
   &=&   \max_{S \sse A} \  \int_S   \mu   \caK_F (a) - e^{\ee} \; \hmu \caK_F (a)  \DX{a}               
     \qq \text{where $\ee \ge 0$.}  
\end{eqnarray*}
     
Instead of distributions we may also use random variables $\caD$ with $ \caD \sim \mu$ and define
the distance betweren these random variables with respect to $F$ by
$ \DD_{F,\caD, \hcaD} \dea \DD_{F,\mu, \hmu}$.
}\end{definition}

This distance is also known as $\alpha$--divergence  with $\aa=e^\ee$, 
see for example \cite{BBG18}.

\FNS{
\begin{definition}{\rm \wfinw{Statistical Privacy}\cite{BR25}\label{SP}\\
    A query $F$ on database $\caD$ is called \bfinw{$(\ee, \dd)$-statistical private}
    if for all  $j \in [1\ldots n]$ and $\alpha, \beta \in D$ holds
    \[ \DD_{F, \caD_{j,\alpha},\caD_{j,\beta}}(\ee) \kla \dd \]
}\end{definition}
}

\begin{definition}{\rm\wfinw{Statistical Privacy (SP)}\cite{BR25}\\
A privacy technique $M$ achieves  $(\ee,\dd)$-statistical privacy 
with respect to a distribution $\mu$ 
(or a collection of distributions) and a query $F$ 
if (for all $\mu$) for every subset $S \sse A$, every entry $j$ and 
all $v,w  \in W$  it  holds for $\caD \sim \mu$:
\[  \Prob{M(F, \caD_{j,v}) \in S}{} \ \kla  e^\ee \ \Prob{M(F,\caD_{j,w}) \in S}{}   + \delta  \ .   \]

The statistical privacy curve  with respect to $\mu$ and $F$ generated by $M$ is the function 
$\Phi_{\mu,F,M}$, where $\Phi_{\mu,F,M}(\ee)$  denote the smallest $\dd$ such that $M$ 
for $F$ and $\mu$ achieves $(\ee,\dd)$-statistical privacy.
In case that no operator $M$ is used 
we simply write $\Phi_{\mu,F}$ --
thus privacy is only generated by the distribution $\mu$.
}\end{definition}

\FNS{
 \begin{definition}{\rm \bfinw{simple-$F$-divergent}\\
     Let $F$ be a query, we call a database $\caD \sim \mu$ \emph{simple-$F$-divergent}
     with respect to $F$ if for all databases $\caD' \sim \hmu$ build from entry
     distributions of $\caD$ of size $k > 1$ it holds that for all $\ee \ge 0$, $0
     < j \leq k$ and $\alpha, \beta$  there exists an $\hat{a} \in \Rb$ such that
     the set 
     \begin{equation*}
         S' = \Argmax{S \subseteq A} \int_S  \mu_{j,\alpha}' \caK_{F}(a) - e^\ee \mu_{j,\beta}' \caK_{F}(a)  \DX{a}
     \end{equation*}
     is equal to $(\infty, \hat{a}]$, $[\hat{a},\infty)$ or $\emptyset$. 
  }\end{definition}
}

In the following we assume that $W$ is completely ordered and consider the partial order on $W^n$
defined by $D \le D'$ iff for all $i$ holds $D_i \le D'_i$.
The answer set $A$ for  queries is restricted to real numbers and a query $F$ is monotone if
$D \le D'$ implies $F(D) \le F(D')$.

\section{Subsampling}\label{section:subsampling_techniques}

Taking a sample from a groundset can be done in different ways -- a fixed size
sample drawn with or without replacement, or with varying size as in Poisson
sampling. Given a function $F$ with $n$ arguments to estimate $F(x_1,\dots,x_n)$
using a sample $y_1,\dots,y_m$ one needs a family of functions $\caF = (F_m(y_1, \dots, y_m))_{m \in
\N}$ that approximate $F$. As already discussed above this only make sense for
symmetric functions, in particular since a sample may occur in an arbitrary
order. Thus in the following when we talk about a query $F$ we actually mean
the family $\caF$ and depending of the size of the sample the appropriate member. 
Classical examples for such families are the mean or the median of a sequence of entries.
On the other hand, functions like min or max are not reasonable candidates 
for sampling unless special restrictions are put on the distributions.
If a distribution is not smooth, but there are very few entries that 
determine the extrema it is unlikely that such an element is drawn
and therefore one cannot hope for an acceptable approximation.

\begin{definition}{\rm \wfinw{Subsampling}
\label{definition:subsampling_technique}\\
   Given a sequence $X=x_1,\ldots,x_n$ of fixed length $n$,
   a sample $Y=y_1,\ldots,y_m$ of $X$  where $y_j = x_{i_j}$ can be described by the sequence of
    indices $i_j$   called a \emph{sampling template} $\tau$.
    Let $\mathfrak{C}^m$ denote the set of all such templates of  size $m \ge 0$ and 
    $\mathfrak{C}$ the union of all $\mathfrak{C}^m$.
 The subset of $\mathfrak{C}$ where an entry $j$ is drawn exactly $k$ times is denoted by 
$\mathfrak{C}_{j,k} \dea \left\{ \tau \in \mathfrak{C} \mid \#\{ i \mid \tau_i =
j \} = k \right\} $   and $\mathfrak{C}_{j,+} \dea \cup_{k = 1}^\infty \mathfrak{C}_{j,k}$.

For templates $\tau,\htau$ of identical size the relation $\tau \approx_j \htau$ is defined by the condition:
$\tau_i \ne j$ implies $\tau_i = \htau_i$ and    $\tau_i = j$ implies $\htau_i \ne j$.
 
 
For fixed length $n$, each specific sampling technique  corresponds to a 
random variable $\caT \sim \mu_\caT$ on $\mathfrak{C}$.
Conditioning  $\caT$ on the  number of times a specific entry is drawn will be denoted by  
{$\caT_{j,k} \sim \mu_\caT \mid   \mathfrak{C}_{j,k}$} and conditioning $\caT$
on the size of subsampled database will be denoted by {$\caT^m \sim \mu_\caT \mid \mathfrak{C}^{m}$.} 
Furthermore conditioning $\caT$ on the event that a
specific entry $j$ is drawn at least once will be denoted by 
$\caT_{j,+} \sim \mu_\caT \mid  \mathfrak{C}_{j,+}$ 
and the negation by {$\caT_{j,-} \sim \mu_\caT \mid \mathfrak{C}_{j,0}$}. 
If  $j$ is clear from the context we will simply write $\caT_+$ and $\caT_-$.
    
    
 For a template $\tau$  of length $m$, the operator   $\SAMP_\tau: W^n \to W^m$ maps 
 databases of size $n$ to databases of size $m$, which  can de described by a Markov kernel $\caK_\tau$.
 Combining it with a  sampling technique $\caT$ we get  an operator $\SAMP_\caT$ from $W^n$ to $W^\star$ with kernel
     \[ \caK_{\caT} \gla \int_{\mathfrak{C}} \caK_\tau  \DX{\mu_\caT(\tau)} \ . \]
}\end{definition}

   
%
    
For example, drawing a random subsample of size $m$ without replacement
corresponds to the uniform  distribution on the set of all injective $\tau$ of size $m$.
For Poisson sampling templates of all sizes between $0$ and $n$ can occur.
For each size $m$ all injective templates have the same probability.

\begin{lemma}
For a a database distribution $\mu$, a query   $F$, a sampling
technique $\caT$ and $S \subseteq A$ it holds
        \[  \mu \; \caK_\T \; \caK_F (S) \gla  \sum_{\tau} \; \mu_\caT (\tau) \; \mu
        \; \caK_\tau \; \caK_F (S) \ .   \]
\end{lemma}

\begin{proof}
    \KURZ{    \begin{align*}
        &\mu \; \caK_\T \; \caK_F (S) =  \int_{W^\ast} \caK_F(S,w) \DX{\mu \caK_\caT (w)}
 \gla \int_{W^\ast} \caK_\caT \caK_F(S,w) \DX{\mu (w)}  \\
        &= \int_{W^\ast} \int_{\mathfrak{C}} \caK_\tau \caK_F(S,w) \DX{\mu_\caT(\tau)} \DX{\mu (w)} 
     \gla \int_{\mathfrak{C}} \int_{W^\ast}  \caK_\tau \caK_F(S,w)  \DX{\mu (w)} \DX{\mu_\caT(\tau)} \\
        &= \int_{\mathfrak{C}} \mu \caK_\tau \caK_F(S) \DX{\mu_\caT(\tau)} 
        \gla\sum_{\tau} \; \mu_\caT (\tau) \; \mu \; \caK_\tau \; \caK_F (S) \ .
    \end{align*}}{    \begin{align*}
        \mu \; \caK_\T \; \caK_F (S) &=  \int_{W^\ast} \caK_F(S,w) \DX{\mu \caK_\caT (w)}
 \gla \int_{W^\ast} \caK_\caT \caK_F(S,w) \DX{\mu (w)}  \\
        &= \int_{W^\ast} \int_{\mathfrak{C}} \caK_\tau \caK_F(S,w) \DX{\mu_\caT(\tau)} \DX{\mu (w)} 
     \gla \int_{\mathfrak{C}} \int_{W^\ast}  \caK_\tau \caK_F(S,w)  \DX{\mu (w)} \DX{\mu_\caT(\tau)} \\
        &= \int_{\mathfrak{C}} \mu \caK_\tau \caK_F(S) \DX{\mu_\caT(\tau)} 
        \gla\sum_{\tau} \; \mu_\caT (\tau) \; \mu \; \caK_\tau \; \caK_F (S) \ .
    \end{align*}}

\end{proof}

\begin{lemma}\label{lemma:bound_by_coupling}
    For databases $\caD, \hcaD$, 
    subsampling techniques  $\SAMP_\T, \SAMP_{\T'}$ and a coupling
    $\nu_{\T,\T'}$ of $\T, \T'$ it holds
    \begin{eqnarray*}
        \DD_{F,\SAMP_\T(\caD),\SAMP_{\T'}(\hat{\caD})}(\ee)  
        &\le& \sum_{\tau, \tau'} \nu_{\T,\T'}(\tau,\tau') \ 
        \DD_{F,\SAMP_\tau(\caD),\SAMP_{\tau'}(\hat{\caD})}(\ee) 
    \end{eqnarray*}
\end{lemma}

\begin{proof}   
    This can be calculated directly by
    \KURZ{    \begin{eqnarray*}
        &&\DD_{F,\SAMP_\T(\caD),\SAMP_{\T'}(\hat{\caD})}(\ee)
        = \int_A \max ( \; 0, \; \mu \caK_{\T} \caK_{F}(a) - e^\ee \; \hat{\mu}
        \caK_{\T'} \caK_{F} (a) )  \DX{a}   \\
        &=& \int_A \max \left( \; 0, \; \sum_{\tau, \tau'} \nu_{\T,\T'}(\tau,\tau')  \mu
        \caK_{\tau} \caK_{F}(a) - e^\ee \cdot \sum_{\tau, \tau'} \nu_{\T,\T'}(\tau,\tau') \hat{\mu}
        \caK_{\tau'} \caK_{F} (a) \right) \DX{a}  \\
        &\leq& \sum_{\tau, \tau'} \nu_{\T,\T'}(\tau,\tau')  \int_A \max \left( \; 0, \; \mu
        \caK_{\tau} \caK_{F}(a) - e^\ee \cdot \hat{\mu} \caK_{\tau'} \caK_{F} (a) \right)  \DX{a}   \\
        &=& \sum_{\tau, \tau'} \nu_{\T,\T'}(\tau,\tau') \DD_{F,\SAMP_\tau(\caD),\SAMP_{\tau'}(\hat{\caD})}(\ee) 
    \end{eqnarray*}}{    \begin{eqnarray*}
        &&\DD_{F,\SAMP_\T(\caD),\SAMP_{\T'}(\hat{\caD})}(\ee) \\
        &=& \int_A \mu \caK_{\T} \caK_{F} (a) \cdot
        \max ( \; 0, \; 1-\exp(\ee-\caL_{\mu \caK_{\T} \caK_{F},\hat{\mu} \caK_{\T'}
        \caK_{F}}(a))) \DX{a}   \\
        &=& \int_A \max ( \; 0, \; \mu \caK_{\T} \caK_{F}(a) - e^\ee \; \hat{\mu}
        \caK_{\T'} \caK_{F} (a) )  \DX{a}   \\
        &=& \int_A \max \left( \; 0, \; \sum_{\tau, \tau'} \nu_{\T,\T'}(\tau,\tau')  \mu
        \caK_{\tau} \caK_{F}(a) - e^\ee \cdot \sum_{\tau, \tau'} \nu_{\T,\T'}(\tau,\tau') \hat{\mu}
        \caK_{\tau'} \caK_{F} (a) \right) \DX{a}  \\
        &\leq& \sum_{\tau, \tau'} \nu_{\T,\T'}(\tau,\tau')  \int_A \max \left( \; 0, \; \mu
        \caK_{\tau} \caK_{F}(a) - e^\ee \cdot \hat{\mu} \caK_{\tau'} \caK_{F} (a) \right)  \DX{a}   \\
        &=& \sum_{\tau, \tau'} \nu_{\T,\T'}(\tau,\tau') \DD_{F,\SAMP_\tau(\caD),\SAMP_{\tau'}(\hat{\caD})}(\ee) 
    \end{eqnarray*}}
\end{proof}

In the SP setting privacy is generated by the uncertainty of an adversary 
about the exact values of the entries.
When taking a random sample of entries the sensitive entry may occur more
than once in case of drawing with replacement or the size of the
sample may vary in case of Poisson sampling.
But smaller size samples have less entropy.
%
Therefore we have to modify the notion of privacy curve by taking the expectation 
over all possible samples. 
This is legitimate since an adversary does not know which sample is actually drawn.

\begin{definition}{\rm \bfinw{Sampling Privacy Curve (\MSPP)}}\label{definition:mspp}\\
For a database $\caD$, a query $F$ and a sampling technique 
$\caT$ the sampling privacy curves are defined by
    \begin{eqnarray*}
       \MSPP_{F,\caD,\T}^j(\ee) &:=& \E_{\tau \sim \T_{j,+}}
       \left[    \Max{v,w \in W} \DD_{F, \ \SAMP_\tau(\Djv),\ \SAMP_{\tau} (\Djw)}(\ee) \right] \ , \\
        \MSPP_{F,\caD,\T}(\ee) &:=&     \Max{j} \ \MSPP_{F,\caD,\T}^j (\ee) \;.
     \end{eqnarray*}
\end{definition}

\vspace*{2ex}

\begin{lemma}\label{lemma:bound_by_mspp}
For $\caD,F,  \nu_{\T,  \hat{\T}}$ as above and $v,w \in \supp(\caD_j)$ it holds
\begin{eqnarray*}
    \DD_{F, \ \SAMP_\T(\caD_{j,v}),  \ \SAMP_{\hat{\T}} (\caD_{j,w})}(\ee) &\leq& \MSPP_{F,\caD,\nu_{\T, \hat{\T}}}^j(\ee) \ .
\end{eqnarray*}
    \end{lemma}

\begin{proof}
Let $\caD \sim \mu$, this is a direct consequence of Lemma
\ref{lemma:bound_by_coupling} and the $\MSPP$ definition
\KURZ{}{
    \begin{eqnarray*}
    &&   \DD_{F, \ \SAMP_\T(\caD_{j,v}),  \ \SAMP_{\hat{\T}} (\caD_{j,w})}(\ee) \\
     &\leq&\E_{(\tau, \hat{\tau}) \sim \nu{\T, \hat{\T}}} \left[ \DD_{F, \ \SAMP_\tau(\caD_{j,v}),  \ 
           \SAMP_{\hat{\tau}} (\caD_{j,w})}(\ee) \right]  \kla \MSPP_{F,\caD,\nu_{\T, \hat{\T}}}^j(\ee) \ . 
\end{eqnarray*}
}
    \end{proof}

   \begin{lemma}\label{lemma:delta_mittelwertsatz}
    For a distribution $\mu$, a query $F$ and $v \in \supp(\omuj)$ it holds 
  \[    \DD_{F,\mujv, \mu}(\ee) \kla \Max{w \in \supp( \omuj)} \DD_{F,\mujv, \mujw}(\ee) \ .      \]
   \end{lemma}

 \begin{proof}
For the density function  $\mu \caK_F$  of $F(\caD)$ the law of total
probability gives
\[         \mu \caK_F(a) \gla \int_{W} \mu_{j,w}\caK_F(a) \; \cdot \; \ovl{\mu_j}(w) \; \DX{w} \ .   \]
Thus,
\KURZ{
    \begin{eqnarray*}
        &&\DD_{F,\mujv, \mu}(\ee) 
        = \int_A \max \left( \; 0, \; \mu_{j,v}\caK_F(a)   - e^\ee  \; \mu \caK_F (a) \right) \ \DX{a} \\
        &=& \int_A \max \left( \; 0, \; \int_{W}  \ovl{\mu_j}(w) \; \mu_{j,v}\caK_F(a) \; \DX{w} -  e^{\ee}
        \int_{W} \mu_{j,w}\caK_F(a) \; \cdot \; \ovl{\mu_j}(w) \; \DX{w} \; \right) \DX{a} \\
        &\le& \int_{W} \ovl{\mu_j}(w) \int_A \max \left( \; 0, \;  \mu_{j,v}\caK_F(a) -  e^{\ee} \; \mu_{j,w}\caK_F(a)  \right) 
             \DX{a}  \;\DX{w} \\
        &\le& \Max{w \in \supp( \omuj)} \underbrace{\left\{  \int_A \max \left( \; 0, \; \mu_{j,v}\caK_F(a) -  e^{\ee} \; \mu_{j,w}\caK_F(a) \right)
           \DX{a}  \right\}}_{ =  \DD_{F,\mujv, \mujw}(\ee)} \ \underbrace{\int_{W}  \omuj (w) \; \DX{w}}_{=1} \ . \\
     \end{eqnarray*}
}{
    \begin{eqnarray*}
        &&\DD_{F,\mujv, \mu}(\ee) \\
        &=& \int_A \max \left( \; 0, \; \mu_{j,v}\caK_F(a)   - e^\ee  \; \mu \caK_F (a) \right) \ \DX{a} \\
        &=& \int_A \max \left( \; 0, \; \int_{W}  \ovl{\mu_j}(w) \; \mu_{j,v}\caK_F(a) \; \DX{w} -  e^{\ee}
        \int_{W} \mu_{j,w}\caK_F(a) \; \cdot \; \ovl{\mu_j}(w) \; \DX{w} \; \right) \DX{a} \\
        &\le& \int_A \int_{W} \ovl{\mu_j}(w) \; \max \left( \; 0, \;  \mu_{j,v}\caK_F(a) -  e^{\ee} \; \mu_{j,w}\caK_F(a) \right) 
            \DX{w} \;  \DX{a} \\
        &\le& \int_{W} \ovl{\mu_j}(w) \int_A \max \left( \; 0, \;  \mu_{j,v}\caK_F(a) -  e^{\ee} \; \mu_{j,w}\caK_F(a)  \right) 
             \DX{a}  \;\DX{w} \\
        &\le& \Max{w \in \supp( \omuj)} \underbrace{\left\{  \int_A \max \left( \; 0, \; \mu_{j,v}\caK_F(a) -  e^{\ee} \; \mu_{j,w}\caK_F(a) \right)
           \DX{a}  \right\}}_{ =  \DD_{F,\mujv, \mujw}(\ee)} \ \underbrace{\int_{W}  \omuj (w) \; \DX{w}}_{=1} \ . \\
     \end{eqnarray*}
}
   \end{proof}

When changing a single entry in a database, 
drawing with replacement can generate samples
that differ at more than 1 position, 
namely when the critical entry $j$ is drawn several times, let us say $k$ times.
For the following analysis we need a smoothness condition for sampling.

\begin{definition}{\rm  $F$-samplable\\
With respect to a query $F$, a database distribution $\mu$ is called 
\emph{$F$-samplable} if for all sampling
templates $\tau, \htau \in \mathfrak{C}^n$, $v,w \in W$ 
and $\ee \ge 0$ a set $S \sse A$ that maximizes the integral
\[    \int_S  \mujv \caK_\tau \caK_{F}(a) - e^\ee \;
    \mujw \caK_{\htau} \caK_{F}(a)  \ \DX{a}  \]
can be chosen as a half open real interval $(\infty, \hat{a}]$, 
$[\hat{a},\infty)$ or the empty set.
}\end{definition}
    
For example, for arbitrary monotone queries $F$, 
one gets a $F$-samplable distribution
if the individual distributions of the entries are binomial, normal
or Laplace distributed.


\begin{lemma}\label{lemma:delta_extendet_mittelwertsatz}
Let $F$ be a monotone query, $\caD \sim \mu$ a $F$-samplable distribution, 
and $\tau, \htau \in \mathfrak{C}^n$ two sampling templates 
with $\tau \approx_j \htau$.
Then for $v \in W$ holds
    \begin{eqnarray*}
       \DD_{ F, \ \SAMP_\tau(\Djv),\ \SAMP_{\htau} (\caD)}(\ee) &\leq&
       \Max{w \in W} \DD_{F, \ \SAMP_\tau(\Djv),\ \SAMP_{\tau} (\Djw)}(\ee) \ . 
    \end{eqnarray*}
 \end{lemma}

 \begin{proof}
    \KURZ{
        Consider the distributional distance $k=\sd(\mujv \caK_{\tau}, \ \mu
        \caK_{\hat{\tau}})$ , for $k=0$ the $j$--th entry does not affect the privacy
        curve therefore the inequality holds and $k=1$ is handled by the previous
        lemma.
        For $k > 1$ w.l.o.g. consider $w_1,\ldots,w_\ell \in W^\ell$ such that $F_{w_1,\dots, w_\ell}(x_1, \dots, x_m) \dea  F(w_1, \dots,
        w_\ell, x_{\ell+1}, \dots, x_m)$.
        }{
        Let  $k=\sd(\mujv \caK_{\tau}, \ \mu \caK_{\hat{\tau}}) $ denote the distance between the two distributions. 
        For $k=0$ a symmetric query does not give any difference.
        The case $k=1$ is handled by the previous lemma.
        For  $k > 1$ we may assume that the entries where the distributions $\mujv \caK_{\tau}$
            and $\mu \caK_{\hat{\tau}}$ differ are the first $k$ ones. 
            For a sequence $w_1,\ldots,w_\ell \in W^\ell$ let
        \[            F_{w_1,\dots, w_\ell}(x_1, \dots, x_m) \dea  F(w_1, \dots, w_\ell, x_{\ell+1}, \dots, x_m) \ . \]
        }
Let $S' \sse A$ be a set that maximizes
    \begin{eqnarray*}
    \DD_{F, \; \SAMP_\tau(\Djv),\; \SAMP_{\hat{\tau}} (\caD)}(\ee)
    &=& \Max{S \subseteq A} \int_S  \mujv \caK_{\tau}
    \caK_{F}(a)-e^\ee \; \mu \caK_{\hat{\tau}} \caK_{F}(a)  \DX{a}.
    \end{eqnarray*}
Since $\mu$ is $F$-samplable $S'$ is equal to either $(\infty, \hat{a}]$,
$[\hat{a},\infty)$ or $\emptyset$ for some $\hat{a} \in \Rb$. 
In case of $\emptyset$ the inequality obviously holds.
Consider the case  $S' = [\hat{a},\infty)$--  the other one can be dealt with in
an analogous way. 
\KURZ{ To prove the inequality we have to show $\mu \caK_{\hat{\tau}}
\caK_{F}([\hat{a}, \infty)) \geq \mujw \caK_{\tau}
\caK_{F}([\hat{a}, \infty))$}{ To prove the inequality we have to show 
\begin{eqnarray*}
 \int_{\hat{a}}^\infty  \mu \caK_{\hat{\tau}} \caK_{F}(a)  \DX{a}  &\geq& \int_{\hat{a}}^\infty  \mujw \caK_{\tau}
   \caK_{F}(a) \DX{a}.
\end{eqnarray*}}
\KURZ{Since $F$ is monotone there exists a $w \in W$ such that
    $\E [F(\SAMP_{\htau}(\caD))] \geq \E   [F_w(\SAMP_{\htau}(\caD))]$.
    This gives
    \begin{eqnarray*}
        \int_{\hat{a}}^\infty \mu \caK_{\hat{\tau}} \caK_{F}(a)  \DX{a} &=& 
        \int_{W^m} \mu \caK_{\hat{\tau}}(x)
        \mathbf{1}_{[\hat{a},\infty)}(F(x))  \DX{x} \\
        &\geq& \Min{w \in W} \int_{W^m} \mu \caK_{\hat{\tau}}(x)
        \mathbf{1}_{[\hat{a},\infty)}(F_{w}(x))  \DX{x} \ .
        \end{eqnarray*}
}{By using the monotonicity of $F$ one can find a $w \in W$ such that
$\E [F(\SAMP_{\htau}(\caD))] \geq \E   [F_w(\SAMP_{\htau}(\caD))]$.
This gives 
\begin{eqnarray*}
\int_{\hat{a}}^\infty \mu \caK_{\hat{\tau}} \caK_{F}(a)  \DX{a} &=& \int_{W^m}
\mu \caK_{\hat{\tau}}(x) \mathbf{1}_{[\hat{a},\infty)}(F(x))  \DX{x} \\
&\geq& \Min{w \in W} \int_{W^m} \mu \caK_{\hat{\tau}}(x)
\mathbf{1}_{[\hat{a},\infty)}(F_{w}(x))  \DX{x} \ .
\end{eqnarray*}}
\KURZ{By iteration it holds 
\begin{eqnarray*}
    \mu \caK_{\hat{\tau}} \caK_{F}([\hat{a},\infty)) 
    &\geq& \Min{(w_1, \dots, w_k)  } \int_{W^m} \mu \caK_{\hat{\tau}}(x)
    \mathbf{1}_{[\hat{a},\infty)}(F_{w_1,\dots, w_k}(x))  \DX{x} 
\end{eqnarray*}
For a symmetric function $F$ the the vector $(w_1, \dots, w_k) $
that minimizes this integral has identical components.}{ By iterating this argument  we get 
\begin{eqnarray*}
    \mu \caK_{\hat{\tau}} \caK_{F}([\hat{a},\infty)) 
    &\geq& \Min{(w_1, \dots, w_k)  } \int_{W^m} \mu \caK_{\hat{\tau}}(x)
    \mathbf{1}_{[\hat{a},\infty)}(F_{w_1,\dots, w_k}(x))  \DX{x} \ .
\end{eqnarray*}
For a symmetric monotone function $F$ one can assume that the vector $(w_1, \dots, w_k) $
that minimizes this integral has identical components. Thus,
\begin{eqnarray*}
    \mu \caK_{\hat{\tau}} \caK_{F}([\hat{a},\infty))
&\geq& \Min{w \in W } \int_{W^m} \mu \caK_{\hat{\tau}}(x)
        \mathbf{1}_{[\hat{a},\infty)}(F_{w,\ldots,w}(x))  \DX{x}    \\
&=&  \Min{w \in W } \mujw \caK_{\tau} \caK_{F}([\hat{a},\infty)).
\end{eqnarray*}}
This can be used to derive the following bound 
\KURZ{
    \begin{eqnarray*}
        \DD_{F, \SAMP_\tau(\Djv),\SAMP_{\hat{\tau}} (\caD)}(\ee)
            &\leq& \Max{w \in W}  \int_{S'}  \mujv \caK_{\tau}
                    \caK_{F}(a)-e^\ee \;  \mujw \caK_{\tau}   \caK_{F}(a)  \DX{a} \\
            &\leq& \Max{w \in W} \; \Max{S \subseteq A}  \int_{S}  \mujv \caK_{\tau}
               \caK_{F}(a)-e^\ee \; \mujw \caK_{\tau}  \caK_{F}(a)  \DX{a} \\
            &=& \Max{w \in W} \; \DD_{F, \SAMP_\tau(\Djv),\SAMP_{\tau} (\Djw)}(\ee). \ .
        \end{eqnarray*}
}{
    \begin{eqnarray*}
        \DD_{F, \SAMP_\tau(\Djv),\SAMP_{\hat{\tau}} (\caD)}(\ee)
            &=&  \int_{S'}  \mujv \caK_{\tau}   \caK_{F}(a)-e^\ee \; \mu \caK_{\hat{\tau}} \caK_{F}  \DX{a} \\
            &\leq& \Max{w \in W}  \int_{S'}  \mujv \caK_{\tau}
                    \caK_{F}(a)-e^\ee \;  \mujw \caK_{\tau}   \caK_{F}(a)  \DX{a} \\
            &\leq& \Max{w \in W} \; \Max{S \subseteq A}  \int_{S}  \mujv \caK_{\tau}
               \caK_{F}(a)-e^\ee \; \mujw \caK_{\tau}  \caK_{F}(a)  \DX{a} \\
            &=& \Max{w \in W} \; \DD_{F, \SAMP_\tau(\Djv),\SAMP_{\tau} (\Djw)}(\ee). \ .
        \end{eqnarray*}
}

 \end{proof}


\section{Privacy Amplification by Subsampling}\label{section:subsampling_results}

\KURZ{Serveral papers have discussed how subsampling can amplify privacy parameters in
the DP model (see among others \cite{KLN11,LQS12,BBG18}). Where the effect on
privacy for a single "worst case" database is considered. 
On the contrary, SP considers a distribution over possible databases, 
thus the interaction between the sampling distribution and the database
distribution has to be analyzed. It turns out that sampling techniques where entries can
appear multiple times are technically more difficult to handle because of possible
dependencies in the sample. }{Serveral papers have discussed how subsampling can amplify privacy parameters in
the DP model (see among others \cite{KLN11,BBG18}). Subsampling with rate $\ll$
turns a $(\ee, \delta)$-DP technique into one with parameters
$\left(\klog(1+\lambda (e^\ee -1)) ,\lambda \; \delta \right)$. 
Here a single database is considered.
On the contrary, SP considers a   distribution on database, 
thus the interaction between a distribution and sampling 
from a random database has to be analyzed. 
It turns out that sampling techniques where entries can
appear multiple times are technically more difficult to handle because of possible
dependencies in the sample. }

\subsection{Sampling without Replacement}

Sampling without replacement has been investigated for differential privacy 
in several papers, among others \cite{KLN11,LQS12,BBG18}.
By extending the coupling technique of Balle et al.~to our kernel characterization of
statistical privacy we can show

\begin{theorem}\label{theorem:subsampling_without_rep} \rm
Let $\mu$ be a distribution for databases of size $n$, $F$ a query and $\T_{n,m}$ sampling
without replacement  with sample size $m$. 
Then  
$F \circ SAMP_{\T_{n,m}}$     achieves
  $(\klog(1+\frac{m}{n} \;  (e^\ee -1)) ,\frac{m}{n} \; \MSPP_{F, \mu, \T_{n,m}}(\ee) )$--statistical privacy
    for $\mu$.
\end{theorem}


\begin{proof}
    Let $\ee' = \log(1+ \frac{m}{n} (e^\ee -1))$ and $v,w \in W$.
    The total variation of $\mujv \caK_\T$ and $\mujw \caK_\T$ is at most $m/n$ -
    the probability that the sensitive entry $j$ is drawn. Thus there exist
    a maximal coupling with parameter $m/n$ such that
    \begin{eqnarray*}
    \mujv \caK_\T &\gla& \left(1-\frac{m}{n}\right) \; \mujv \caK_{\T_{j,-}} + \frac{m}{n} \; \mujv \caK_{\T_{j,+}} \\
    \mujw \caK_\T &\gla& \left(1-\frac{m}{n}\right) \; \mujv \caK_{\T_{j,-}} +\frac{m}{n} \; \mujw \caK_{\T_{j,+}} \ .
    \end{eqnarray*}
Because of the advanced joint convexity property it
holds 
\begin{align*}
    &\DD_{F,\mujv \caK_\T, \mujw \caK_\T}(\ee') \\ 
    &\kla \frac{m}{n} \left((1-e^{\ee'-\ee}) \; \DD_{F, \mujv \caK_{\T_{j,+}} ,
    \mujv \caK_{\T_{j,-}}}(\ee) \; + \; e^{\ee'-\ee} \; \DD_{F, \mujv \caK_{\T_{j,+}},
    \mujw \caK_{\T_{j,+}}}(\ee)  \right).
\end{align*}
By definition we have that 
\[    \DD_{F, \mujv \caK_{\T_{j,+}},\mujw \caK_{\T_{j,+}}}(\ee) \kla  
    \MSPP_{F, \caD, \T_{j,+}}^j(\ee) \gla \MSPP_{F, \caD, \T}^j(\ee) \ .\]
\KURZ{Consider a coupling $\nu_{\T_{j,+}, \T_{j,-}}$ of $\T_{j,+}$ and
$\T_{j,-}$ which matches all subsamples equal except where $\T_{j,+}$ places the
sensitive entry. At these locations $\T_{j,-}$ selects an entry not drawn yet
uniformly distributed. Thus for all $(\tau_+, \tau_-) \in \supp \;\nu_{\T_{j,+},
\T_{j,-}}$ one can apply Lemma \ref{lemma:delta_mittelwertsatz} to the sampled
databases. Doing this after using Lemma \ref{lemma:bound_by_coupling} we can
bound the privacy curve as follows:
}{Consider a coupling $\nu_{\T_{j,+}, \T_{j,-}}$ of $\T_{j,+}$ and $\T_{j,-}$ such that $\tau_+
\approx_j \tau_-$ holds for $(\tau_+, \tau_-) \in \supp \;\nu_{\T_{j,+}, \T_{j,-}}$,
which matches all subsamples equal except where $\T_{j,+}$ places the sensitive
entry. At these locations $\T_{j,-}$ selects an entry not drawn yet uniformly
distributed. Using Lemma \ref{lemma:bound_by_coupling} and Lemma
\ref{lemma:delta_mittelwertsatz} we can bound the privacy curve as follows:}
\begin{eqnarray*}
    \DD_{F, \mujv \caK_{\T_{j,+}} ,\mujv \caK_{\T_{j,-}}}(\ee)
    &\leq& \sum_{\tau_+, \tau_-} \nu_{\T_{j,+},\T_{j,-}}(\tau_+,\tau_-) \;
    \DD_{F, \mujv \caK_{\tau_+} ,\mujv \caK_{\tau_-}}(\ee) \\
    &\leq& \sum_{\tau_+, \tau_-} \nu_{\T_{j,+},\T_{j,-}}(\tau_+,\tau_-) \; \Max{w \in
    D} \DD_{F, \mujv \caK_{\tau_+} ,\mujw \caK_{\tau_+}}(\ee) \\
    &\leq& \MSPP_{F, \caD, \T_{+}}^j (\ee) = \MSPP_{F, \caD, \T}^j(\ee).
\end{eqnarray*}
This implies the claim
\[ \DD_{F, \mujv \caK_{\T},\mujw \caK_{\T}}(\ee) \kla \frac{m}{n} \; \MSPP_{F, \caD,T}^j (\ee) 
    \kla \frac{m}{n} \; \MSPP_{F, \caD, \T} (\ee)  \  .\]
\end{proof}

\vspace{1ex}

\noindent
In case of identically and independently distributed entries Theorem~\ref{theorem:subsampling_without_rep} gives

\begin{corollary}\label{corollary:subsampling_without_rep:iid} \rm
    Let $\omu$ be the distribution of a database entry and $\mu_n=\omu^n$ and $\mu_m=\omu^m$ 
    be the product distributions for databases of size $n$, resp.~$m$.\\
 Then for a query $F$ the sampling mechanism   without replacement  $\T_{n,m}$ achieves
    $(\klog(1+ \frac{m}{n} \; (e^\ee -1)) ,\;  \frac{m}{n}   \;    \Phi_{\mu_m,F}(\ee) )$--statistical privacy 
    with respect to $\mu_n$.
\end{corollary}

 When does subsampling provide an improvement? 
 For databases of size $n$ the
distribution $\mu_n$ guarantees $(\ee, \; \Phi_{\mu_n,F}(e))$-statistical
privacy, while on a subset of size $m$ one gets tuples 
$(\hat{\ee},\; \Phi_{\mu_m,F}(\hat{\ee}))$ for arbitrary $\hat{\ee} \ge 0$. 
For small $\hat{\ee}$ the term  $\klog(1+ \frac{m}{n} \; (e^{\hat{\ee}} -1))$ can be
approximated by $\ll \; \hat{\ee}$, where $\ll=m/n$. 
Thus the new privacy parameters are about
$(\ll \; \hat{\ee}, \; \ll \;  \Phi_{\mu_{\ll n},F}(\hat{\ee}) )$. This is a
linear reduction in both parameters, but the distribution $\mu_{\ll n}$ on the
smaller sample of size $m$ provides less statistical privacy. To make the
$\ee$-parameter equal we set $\hat{\ee}= \ll^{-1} \; \ee$ and get for
subsampling  the pair $(\ee, \; \ll \;  \Phi_{\mu_{\ll n},F}(\ll^{-1} \; \ee))$. 
For  a query $F$ and  distribution $\mu$  subsampling increases privacy iff (approximately)
\begin{equation} 
    \Phi_{\mu,F}(\ee) \gua \frac{m}{n} \; \MSPP_{F, \mu,\T_{n,m}}(\klog(1+\frac{n}{m} \;  (e^\ee -1)))  \ .
 \label{conclusion:sampling_viabiliy_check}
\end{equation}

\begin{figure}[h] 
    \centering
    \includegraphics[width=.80\linewidth]{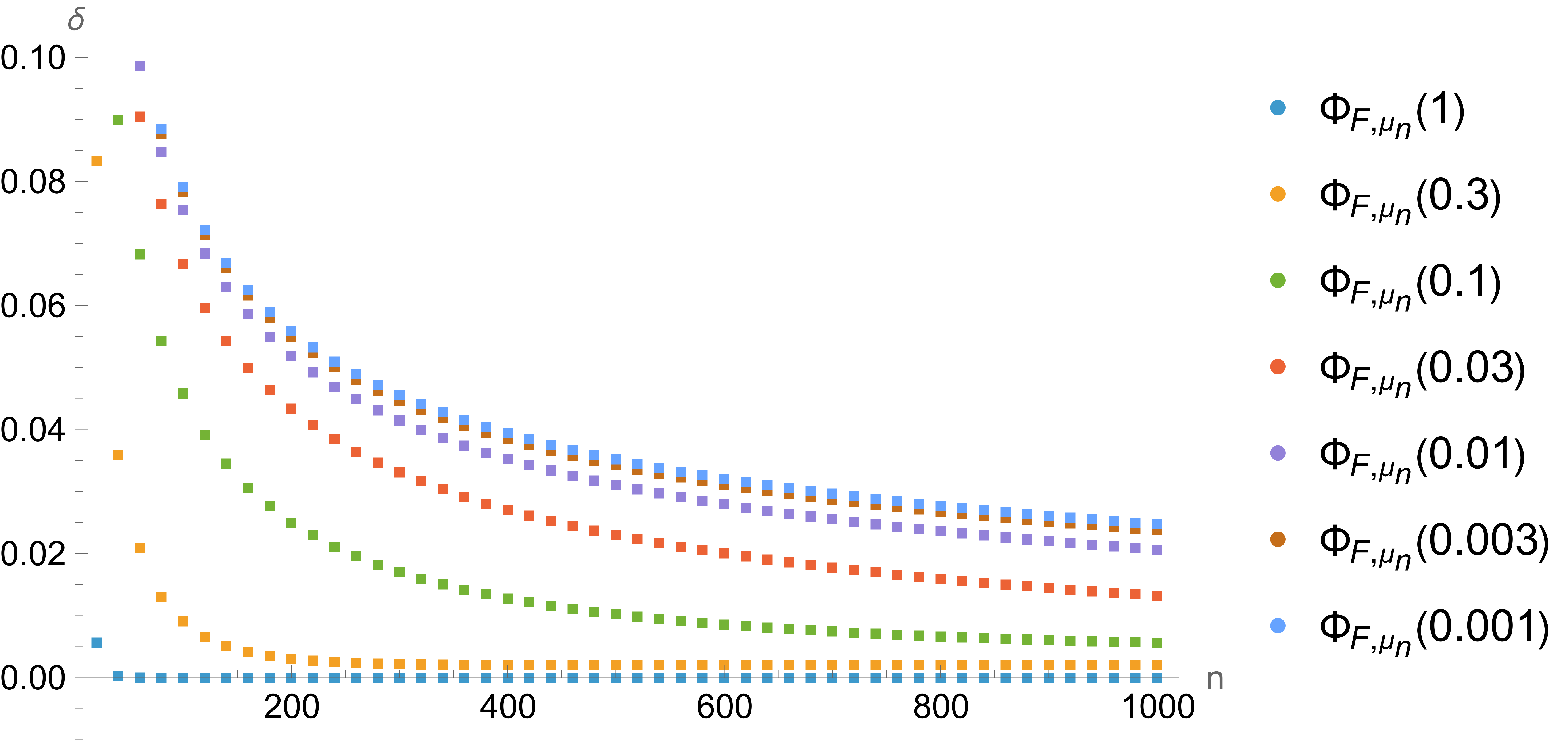}
    \caption{$\dd$ values for  property queries with $p=1/2$ for different $\ee$ values. 
        To ensure that the curves of $\ee = 1$, $\ee = 0.3$, and $\ee = 0.1$ 
        are better differentiated from each other, 
        an upward correction of $0.001$ was made for $\ee = 0.3$
        and $0.002$ for $\ee = 0.1$.}
        \label{figure:PQ_Delta_Parameter_DB_Size}
\end{figure}
Thus, the $\delta$-value decreases if the greater distributional divergence,  
when going from database size $n$ to smaller size $m$, is compensated by the
linear factor $\ll$ and the increase of the margin $\ee$ to $\hat{\ee}=\ll^{-1} \; \ee$. 
Since the entropy of the distribution decreases linear in $\ll$,
applying a function $F$ seems unlikely to give a larger entropy decrease.
In addition the margin $\hat{\ee}$ grows linearly. Since there are no simple
formulas for these relations one has to check the range of improvements for
every $\omu$ and $F$. For large $n$, since $\omu^n$ converges to a normal
distribution improvements should occur for a large range of $\lambda$.

As a typical example consider property queries which ask for
the proportion of database entries that have a certain property that has a priori probability $p$ (see\cite{BR25}). 
Fig.~\ref{figure:PQ_Delta_Parameter_DB_Size} shows the dependency 
between $n$ and $\dd$ for different $\ee$-values.
For larger $n$ the decrease of $\dd$ becomes smaller.

 \KURZ{ By scaling the answer to a property query we get a counting query and
 since SP fulfills the postprocessing condition they are equivalent with respect
 to privacy. 
 These counting queries are binomial distributed and for large $n$ can be approximated by
 a normal distribution $\caN(n \; p, n \; p (1-p))$ if $p$ is not close to $0$ or  $1$. 
 By considering the noise generated by the nonsensitive entries as additive
 noise in the DP sense, where counting queries have a sensitivity of $1$, the DP
 privacy bound in Appendix A of \cite{DR14} allows for the approximation
\[ \Phi_{\mu_n,F}(\ee) \approx 10/(n \; p(1-p) \; \ee^2 ). \] 
Thus for large $m$ and $n$ 
the inequality~(\ref{conclusion:sampling_viabiliy_check}) is equal to 
\begin{eqnarray*}
    & \Phi_{\mu_n,F}(\ee) &\gra \frac{m}{n} \Phi_{\mu_m,F}(\frac{n}{m} \; \ee) \\
    \Longleftrightarrow \qd & \Phi_{\mu_n,F}(\ee) &\gra \frac{m^2}{n^2} 10/(n \; p(1-p) \;
    (\ee)^2 ) \\ 
    \Longleftrightarrow \qd & \Phi_{\mu_n,F}(\ee) &\gra \frac{m^2}{n^2} \Phi_{\mu_n,F}(\ee) \ .
 \end{eqnarray*}
 Since the last inequality always holds for $m<n$, we can conclude that the privacy of property queries 
 is indeed amplified by subsampling.
}{ Since SP fulfills the postprocessing condition, equivalently with respect to privacy, 
 one may consider counting queries that ask for the 
 total number of entries with the property.
 Now the sensitivity is fixed to $1$ instead of decreasing as $1/n$.
The answer to a query is binomial distributed and for great $n$ can be approximated by
the normal distribution  $\caN(n \; p, n \; p (1-p))$ if $p$ is not close to $0$ or $1$.
In the DP sense the values of nonsensitive entries can be considered as noise which 
allows the approximation
\[ \Phi_{\mu_n,F}(\ee) \approx 10/(n \; p(1-p) \; \ee^2 ) \] 
using the estimation in Appendix A of \cite{DR14}. 
Thus, if $m$ is not too small one gets the following conclusion
\begin{eqnarray*}
    & \Phi_{\mu_n,F}(\ee) &\gra \frac{m}{n} \Phi_{\mu_m,F}(\frac{n}{m} \; \ee) \\
    \Longleftrightarrow \qd & \Phi_{\mu_n,F}(\ee) &\gra \frac{m}{n} 10/(m \; p(1-p) \;
    (\frac{n}{m} \; \ee)^2 ) \\
    \Longleftrightarrow \qd & \Phi_{\mu_n,F}(\ee) &\gra \frac{m^2}{n^2} 10/(n \; p(1-p) \;
    (\ee)^2 ) \\ 
    \Longleftrightarrow \qd & \Phi_{\mu_n,F}(\ee) &\gra \frac{m^2}{n^2} \Phi_{\mu_n,F}(\ee) \ .
 \end{eqnarray*}
 The last inequality trivially holds for $m<n$.}
In Fig.~\ref{figure:PQ_Delta_Parameter_DB_Size} one can see that
decreasing the database size $n$ leads to a smaller increase of $\dd$ 
(that means slope larger than~$-1$) except for very small values of $n$.
Fig.~\ref{figure:Ratio_Verification_Comparison} gives a plot of the ratios of the exact values of the
two sides in inequality~(\ref{conclusion:sampling_viabiliy_check}).
\begin{figure}[!h] 
    \centering
    \includegraphics[width=.80\linewidth]{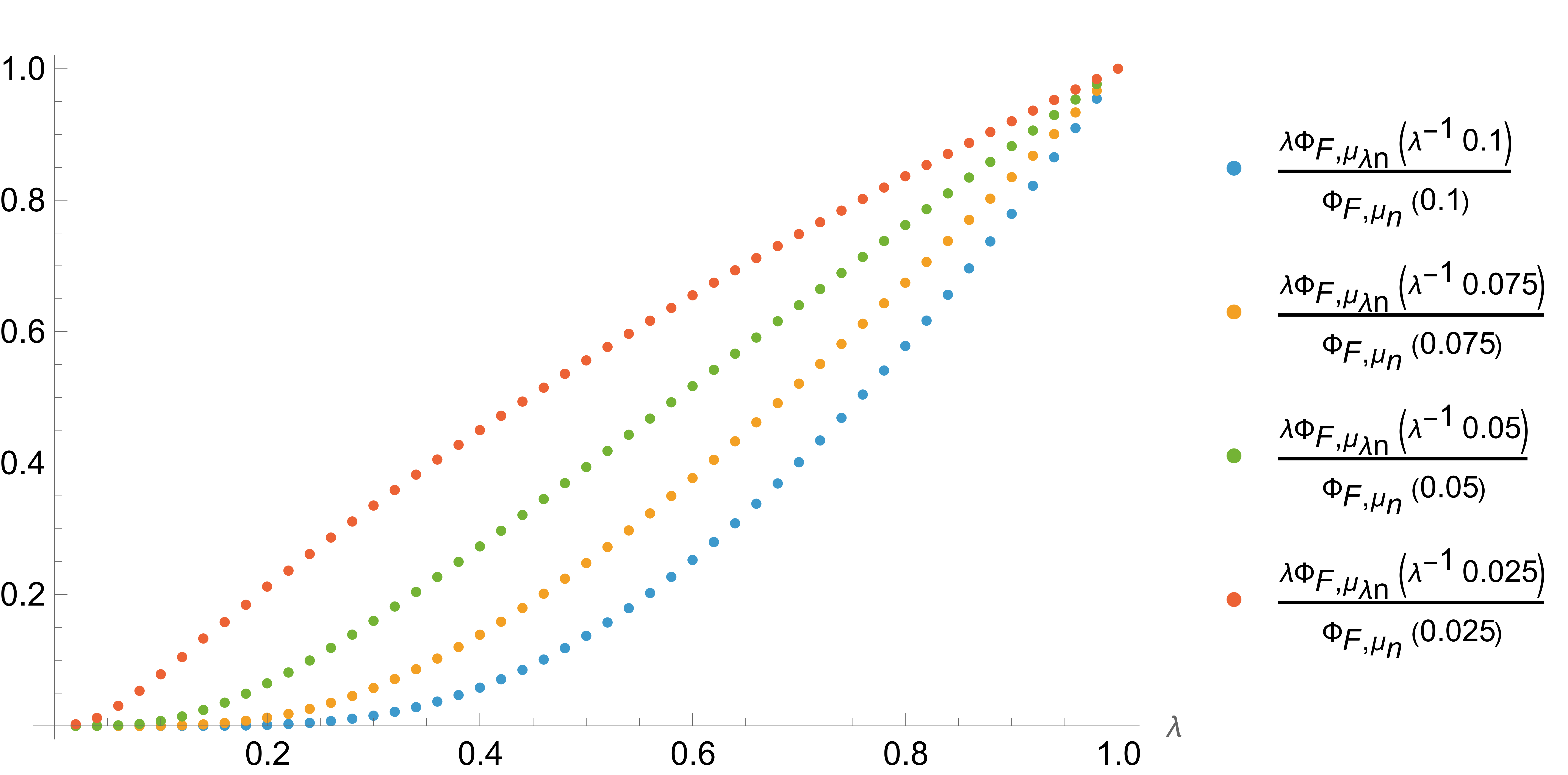}
    \caption{Ratio of the $\dd$ parameter with, resp.~without subsampling
    given a database of size $n=1000$ 
    for property queries with  $p=1/2$ and different sampling rates $\ll$. 
    Here $\ee$ takes the values  $0.1$, $0.075$, $0.05$ and $0.025$.}
    \label{figure:Ratio_Verification_Comparison}
 \end{figure}
It shows that for smaller $\ee$ values the ratio is larger meaning
 subsampling is less effective, but still one gets an improvement.
 
However, subsampling always decreases the utility since when
computing  $F$  on the whole database one gets the exact result  while
depending on the sampling rate $\ll$ and the function $F$ the results on samples
have some variance.

\subsection{Poisson Sampling}
Now consider Poisson sampling where every entry is drawn independently with
the same probability $\ll$ -- thus the sample size may vary.    
We can show two bounds for $\dd$ and conjecture that the second one is never
worse than the first one, but an analytic proof for this claim seems to be
tedious.

\begin{theorem}\label{theorem:subsampling:poisson}
    Let $\mu$ be a database distribution, $F$ a query and $\T_\ll$  Poisson sampling  with rate $\ll$.    
    Then $F \circ \SAMP_{\T_\ll}$ is $(\ee, \dd^\star(\ee))$--statistical
    private for 
    \[
        \dd^\star(\ee) \gla
        \sum_{ =0}^{n} \binom{n}{m} \ll^m (1-\ll)^{n-m} \; \frac{m}{n} \;  \MSPP_{F, \mu, \T_{n,m}} \left(\klog \left(1+\frac{n}{m} \;  (e^\ee -1)\right)\right).
    \]

 \end{theorem}
 \begin{proof}
    Let $v,w \in W$ and $\ee > 0$. Recognize, that the sample
    distributions can be broken down by size $m \in [0:n]$ for $u \in \{v,w\}$, it holds
    \begin{eqnarray}
       \mu_{j,u} \caK_{\T} &=& \sum_{m=0}^{n} \binom{n}{m} \ll^m (1-\ll)^{n-m}
       \mujv \caK_{\T^m} \; , \label{subsampling:poisson:length_split_full}\\
    \end{eqnarray}
    This allows for the following bound for any $\ee > 0$
    \begin{align}
      &  \DD_{F,\mujv \caK_\T, \mujw \caK_\T}(\ee)\\
       &\gla \int_A \max ( \; 0, \; \mujv \caK_{\T} \caK_{F}(a) - e^{\ee} \; \mujw
        \caK_{\T} \caK_{F} (a) )  \DX{a}   \\
        &\kla \sum_{m=0}^{n} \binom{n}{m} \ll^m (1-\ll)^{n-m} \int_A \max ( \; 0, \; \mujv \caK_{\T^m} \caK_{F}(a) - e^{\ee} \; \mujw
        \caK_{\T^m} \caK_{F} (a) )  \DX{a}   \\
        &\gla  \sum_{m=0}^{n} \binom{n}{m} \ll^m (1-\ll)^{n-m} \DD_{F,\mujv \caK_{\T^m}, \mujw \caK_{\T^m}}(\ee)
    \end{align}
    Since $\T^m$ now draws a sample of fixed size $m$ where every entry has the
    same probability to be drawn exactly once this realizes sampling without
    replacement for a sample size of $m$ thus we can apply Theorem
    \ref{theorem:subsampling_without_rep} which bounds $ \DD_{F,\mujv \caK_\T, \mujw \caK_\T}(\ee)$  by 
    \begin{equation*}
 \sum_{m=0}^{n} \binom{n}{m} \ll^m (1-\ll)^{n-m} \; \frac{m}{n} \;  \MSPP_{F, \mu, \T_{n,m}} \left(\klog \left(1+\frac{n}{m} \;  (e^\ee -1)\right)\right).
    \end{equation*}
 \end{proof}
 \noindent
 In the case of i.i.d entries we get:
 
\begin{corollary}\label{corollary:poisson:iid}
    Let $\omu$ be the distribution of a database entry and  $\mu_m=\omu^m$ 
    the product distributions for a database of size $m$.
Then for a query $F$ Poisson sampling with  parameter $\ll$ and
%
    \[
        \dd^\star(\ee) = 
        \sum_{m=0}^{n} \binom{n}{m} \lambda^m (1-\lambda)^{n-m} \; \frac{m}{n} \; \DD_{F,\mu_m}\left(\klog \left(1+\frac{n}{m} \;  (e^\ee -1)\right)\right)
    \]
    $F \circ \SAMP_{\T_\ll}$ achieves $(\ee, \dd^\star(\ee))$--statistical privacy.
 \end{corollary}

 \noindent
 Thus, Poisson sampling with parameter $\ll{}$ improves privacy if
  \[ \Phi_{\mu,F}(\ee) \gua  \; \sum_{m=0}^{n}
   \binom{n}{m} \lambda^m (1-\lambda)^{n-m} \ \frac{m}{n} \; \MSPP_{F, \mu, T^m}(\klog(1+\ll{} \;  (e^\ee -1)))  \ .\] 
   Fig.~\ref{figure:Poisson_sampling} visualizes the effect of Poisson sampling
   for property queries. As for sampling without replacement the improvement by
   Poisson sampling decreases when $\ee$ gets smaller.

\begin{figure}[h] 
    \centering  \includegraphics[width=.80\linewidth]{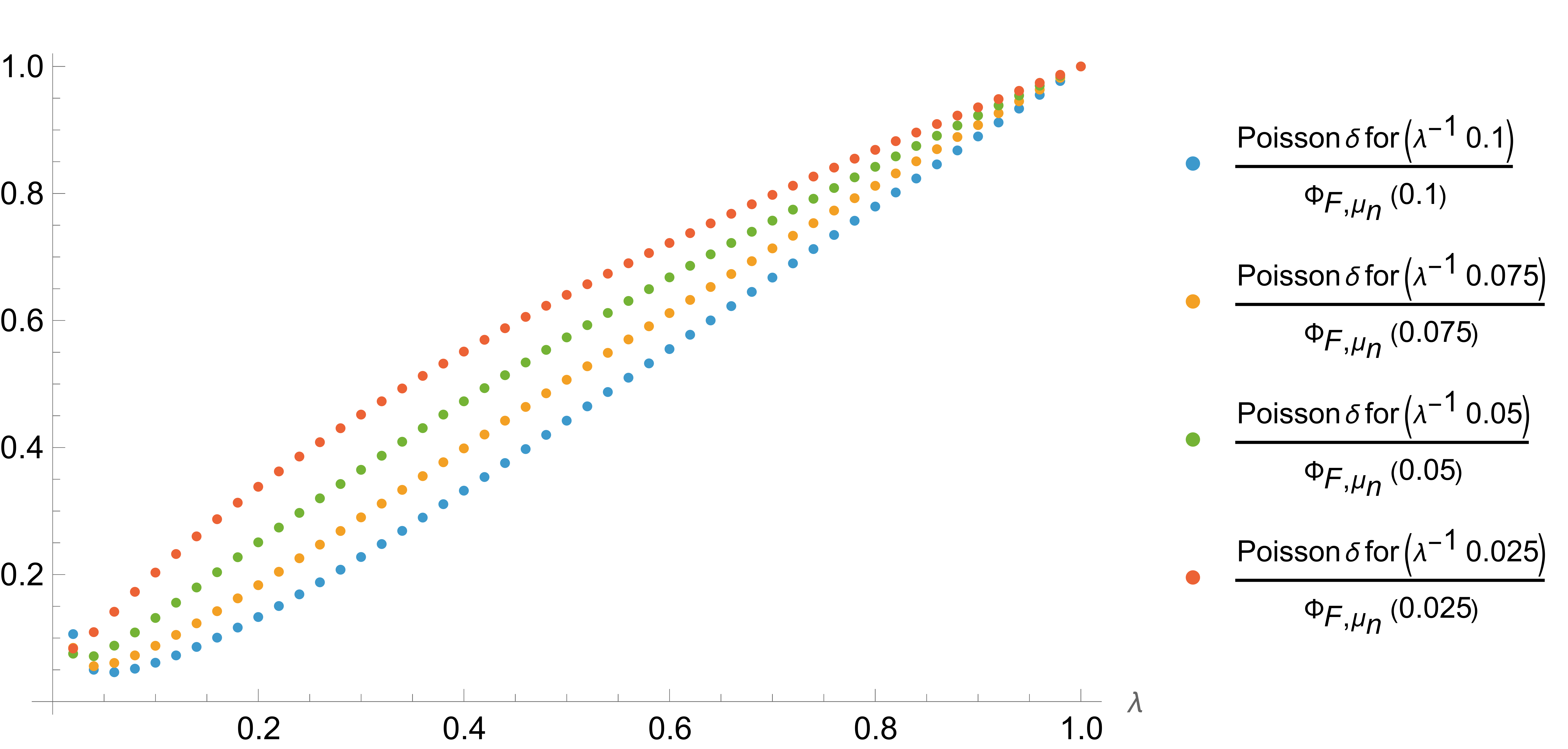}
    \caption{Ratio of the $\dd$ parameter for Poisson subsampling
    given a database of size $n=100$ 
    for property queries with  $p=1/2$ and different sampling rates $\ll$. 
    Here $\ee$ takes the values  $0.1, 0.075, 0.05, 0.025$.}
    \label{figure:Poisson_sampling}
 \end{figure}\vspace*{-2ex}

\subsection{Sampling with Replacement}

Sampling with replacement turns out to be technically most difficult to analyse
with respect to statistical privacy. The main problem here is the dependency
when a critical entry is drawn more than once. For this reason we need the
technical condition that the distribution is $F$-samplable. Let $\T_k$ denote
sampling where the sensitive entry is drawn $k$ times.

 \begin{theorem}
Let $F$ be a monotone query, $\mu$ a $F$-samplable distribution 
and $\T^+_{n,m}$  sampling with replacement genearating a sample of size $m$
from a set of size $n$. 
Then $F \circ \SAMP_{\T^+_{n,m}}$ achieves  
\[ ( \rlog(1+(1-(1-1/n)^m) (e^\ee -1)) ,\; \frac{n}{m}  \; \sum_{k=1}^{m}
    \binom{n}{k} (1/n)^k (1-1/n)^{(m-k)} \;  \MSPP_{F, \caD, \T_k}(\ee) ) \] 
statistical privacy.
  \end{theorem}
  \begin{proof}
    \KURZ{ Let $\ll = m/n$ and $\ee' = \log(1+\lambda (e^\ee -1))$, $j \in
    [1:n]$ and $v, w \in W$. To ease the notation we simply write $\T$ for
    $\T^+_{n,m}$. As for sampling without replacement we consider a maximal
    coupling matching the distributions where the sensitive entry is drawn. Here
    the total variation of $\mujv \caK_\T$ and $\mujw \caK_\T$ is at most $\ll =
    1-(1-1/n)^m$ the probability to not draw the sensitive element. Using the
    advanced joint convexity property we bound $\DD_{F,\mujv \caK_\T, \mujw
    \caK_\T}(\ee') $ by 
    \begin{equation}
  \lambda \left((1-e^{\ee'-\ee}) \; \DD_{F, \mujv \caK_{\T_{j,+}} ,
        \mujv \caK_{\T_{j,-}}}(\ee) \; + \; e^{\ee'-\ee} \; \DD_{F, \mujv \caK_{\T_{j,+}},
        \mujw \caK_{\T_{j,+}}}(\ee)  \right). \label{subsampling:with_replacement:ajc}
    \end{equation}
Let us  first bound the second summand of 
\eqref{subsampling:with_replacement:ajc} by applying 
Lemma~\ref{lemma:bound_by_coupling}.
\begin{eqnarray*}
        \DD_{F, \mujv \caK_{\T_{j,+}}, \mujw \caK_{\T_{j,+}}}(\ee) 
        &\kla& \MSPP_{F, \caD ,\T_{j,+}}^j (\ee) 
        \gla \sum_{k \geq 0} \nu_{\T_{j,+}} (\mathfrak{C}_{j,k}) \; \MSPP_{F,
        \caD ,\T_{j,k}}^j (\ee).
        \label{subsampling:with_replacement:sum}
\end{eqnarray*}
To bound the first summand consider the coupling $\nu_{\T_{j,+}, \T_{j,-}}$ that
maps the templates $\tau_{+}, \tau_{-}$ such that $\tau_{+} \approx_j \tau_{-}$.
This means they map the same entry if the entry is not the sensitive one. Else
let $I$ be the set of indices where $(\tau_{+})_k =j$. Then for each $z \in I$
the entry $(\tau_{-})_z$ is uniformly distributed on $\{ 1, \ldots, n
\}\backslash\{ j \}$ and independent of the distributions of the entries
$(\tau_{-})_v$ $v \in I\backslash\{z\}$. That coupling allows for the application of Lemma
\ref{lemma:delta_extendet_mittelwertsatz} after bounding with Lemma
\ref{lemma:bound_by_mspp} thus
\begin{eqnarray*}
  \DD_{F, \mujv \caK_{\T_{j,+}} ,\mujv \caK_{\T_{j,-}}}(\ee)  
   &\kla&  \sum_{\tau_{+}, \tau_{-}} \nu_{\T_{j,+}, \T_{j,-}}(\tau_{+},\tau_{-})
         \DD_{F, \mujv \caK_{\tau_+} ,\mujv \caK_{\tau_-}}(\ee) \\
   &\kla& \sum_{\tau_{+}, \tau_{-}} \nu_{\T_{j,+},
         \T_{j,-}}(\tau_{+},\tau_{-}) \Max{w \in W} \DD_{F, \mujv \caK_{\tau_+} ,
         \mujw \caK_{\tau_+}}(\ee) \\
   &\kla& \MSPP_{F,\caD,\T}^j (\ee) \gla
     \sum_{k \geq 0} \nu_{\T_{j,+}} (\mathfrak{C}_{j,k}) \; \MSPP_{F, \caD ,\T_{j,k}}^j(\ee)  \ .
     \end{eqnarray*}
Since the probability of drawing a specific entry $k$ times is binomial
     distributed we can calculate $\nu_{\T_{j,+}} (\mathfrak{C}_{j,k})$ giving 
     \begin{eqnarray*}
        \DD_{F,\mujv \caK_\T, \mujw \caK_\T}(\ee') &\kla\lambda\; \sum_{k=1}^{m} \binom{n}{k}
        (1/n)^k (1-1/n)^{(m-k)} \; \MSPP_{F, \caD ,\T_{j,k}}^j (\ee)\\
        &\kla \lambda\; \sum_{k=1}^{m} \binom{n}{k}
     (1/n)^k (1-1/n)^{(m-k)} \; \MSPP_{F, \caD ,\T_{k}} (\ee) \ . 
     \end{eqnarray*}
    }{
        Let $\ll = m/n$ and $\ee' = \log(1+\lambda (e^\ee -1))$, $j \in
        [1:n]$ and $v, w \in W$. 
        To ease the notation we simply write $\T$ for $\T^+_{n,m}$.
        The total variation of $\mujv \caK_\T$ and $\mujw \caK_\T$ is at most $\ll =
        1-(1-1/n)^m$ which matches the probability that we do not draw the sensitive
      at all. This implies that there exist a maximal coupling with parameter \ll. 
        By matching the templates in the coupling with $\T_{j,+}, \T_{j,-}$ as defined
        earlier we get
        \begin{eqnarray*}
           \mujv \caK_\T &\gla& (1-\lambda) \; \mujv \caK_{\T_{j,-}} + \lambda \; \mujv \caK_{\T_{j,+}} \\
           \mujw \caK_\T &\gla& (1-\lambda) \; \mujv \caK_{\T_{j,-}} + \lambda \; \mujw \caK_{\T_{j,+}} 
        \end{eqnarray*}
        Because of the advanced joint convexity property it
        holds 
        \begin{align}
            &\DD_{F,\mujv \caK_\T, \mujw \caK_\T}(\ee') \nonumber \\ 
            &\kla \lambda \left((1-e^{\ee'-\ee}) \; \DD_{F, \mujv \caK_{\T_{j,+}} ,
            \mujv \caK_{\T_{j,-}}}(\ee) \; + \; e^{\ee'-\ee} \; \DD_{F, \mujv \caK_{\T_{j,+}},
            \mujw \caK_{\T_{j,+}}}(\ee)  \right). \label{subsampling:with_replacement:ajc}
        \end{align}
    Let us  first bound the second summand of the right hand side of 
    \eqref{subsampling:with_replacement:ajc} by applying 
    Lemma~\ref{lemma:bound_by_mspp}.
    \begin{eqnarray*}
            \DD_{F, \mujv \caK_{\T_{j,+}}, \mujw \caK_{\T_{j,+}}}(\ee) 
            &\kla& \MSPP_{F, \caD ,\T_{j,+}}^j (\ee) 
            \gla \sum_{k \geq 0} \nu_{\T_{j,+}} (\mathfrak{C}_{j,k}) \; \MSPP_{F,
            \caD ,\T_{j,k}}^j (\ee).
            \label{subsampling:with_replacement:sum}
    \end{eqnarray*}
    To bound the first summand 
    consider the coupling
    $\nu_{\T_{j,+}, \T_{j,-}}$ that maps the templates $\tau_{+}, \tau_{-}$
        such that $\tau_{+} \approx_j \tau_{-}$.
    This means they map the same entry if
    the entry is not the sensitive one.
    Else let $I$ be the set of indices where $(\tau_{+})_k =j$.
     Then for each $z \in I$ the entry $(\tau_{-})_z$ is
        uniformly distributed on $\{ 1, \ldots, n \}\backslash\{ j \}$ and independent
        of the distributions of the entries $(\tau_{-})_v$ $v \in I\backslash\{z\}$.
        This creates a coupling $\nu_{\T_{j,+}, \T_{j,-}}$ for $\T_{j,+}$ and
        $\T_{j,-}$ that fulfills the conditions of Lemma
        \ref{lemma:delta_extendet_mittelwertsatz}. 
    This together with Lemma
        \ref{lemma:bound_by_coupling} implies
    \begin{eqnarray*}
      \DD_{F, \mujv \caK_{\T_{j,+}} ,\mujv \caK_{\T_{j,-}}}(\ee)  
       &\gla&  \sum_{\tau_{+}, \tau_{-}} \nu_{\T_{j,+}, \T_{j,-}}(\tau_{+},\tau_{-})
             \DD_{F, \mujv \caK_{\tau_+} ,\mujv \caK_{\tau_-}}(\ee) \\
       &\kla& \sum_{\tau_{+}, \tau_{-}} \nu_{\T_{j,+},
             \T_{j,-}}(\tau_{+},\tau_{-}) \Max{w \in W} \DD_{F, \mujv \caK_{\tau_+} ,
             \mujw \caK_{\tau_+}}(\ee) \\
       &\kla& \MSPP_{F,\caD,\T}^j (\ee) \gla
         \sum_{k \geq 0} \nu_{\T_{j,+}} (\mathfrak{C}_{j,k}) \; \MSPP_{F, \caD ,\T_{j,k}}^j(\ee)  \ .
         \end{eqnarray*}
    Since the probability of drawing a specific entry $k$ times is binomial
         distributed we can calculate $\nu_{\T_{j,+}} (\mathfrak{C}_{j,k})$ giving 
         \begin{eqnarray*}
            \DD_{F,\mujv \caK_\T, \mujw \caK_\T}(\ee') &\kla\lambda\; \sum_{k=1}^{m} \binom{n}{k}
            (1/n)^k (1-1/n)^{(m-k)} \; \MSPP_{F, \caD ,\T_{j,k}}^j (\ee)\\
            &\kla \lambda\; \sum_{k=1}^{m} \binom{n}{k}
         (1/n)^k (1-1/n)^{(m-k)} \; \MSPP_{F, \caD ,\T_{k}} (\ee) \ . 
         \end{eqnarray*}
    }
  \end{proof}
As before the question arises when does sampling with replacement give an advantage, 
that means when holds
\[ \Phi_{\mu,F}(\ee) \gua \lambda  \; \sum_{k=1}^{m}
  \binom{n}{k} (1/n)^k (1-1/n)^{(m-k)} \;  \MSPP_{F, \mu, \T_k}
  (\klog(1+\frac{n}{m} \;  (e^\ee -1))) \ .  \] 
This condition is significantly more complex than for sampling without replacement.
However, we expect a similar behavior since the probability to draw the critical element several times
 decreases exponentially.

\section{Trade-off Functions for Statistical Privacy}

$\ee$-divergence and $(\ee,\dd)$-curves measure the amount of privacy that  can
be guaranteed. Dong et al.~have proposed an alternative approach  considering
privacy as a hypothesis testing problem \cite{DRS22}.
  
\begin{definition}{\rm \wfinw{Trade-off Function}\label{tradeoff-function}\\
For  random variables $X \sim \mu_x$ and $Y \sim \mu_y$ on the same
 space $A$, the trade-off function $T_{X,Y} \colon [0,1] \to [0,1]$  is defined by
 \begin{equation*}
T_{X,Y}(\alpha) \dea  \Inf{S \subseteq A : \; \Prob {X \not\in S}{} \le \aa}  \ \Prob{Y \in S}{}  \ . 
      \end{equation*}
  }\end{definition}

Now the ability of an adversary to differentiate between two situations where
a critical entry is either $v$ or $w$ is measured by such a function.
This approach can be adapted to statistical privacy as well.

 \begin{definition}{\rm \wfinw{$T$--Statistical Privacy ($T$--SP)}\label{T_SP}\\
For a  distribution $\mu$, a query $F$  and a trade-off  function $T$, 
a mechanism $M$ achieves  \bfinw{$T$-statistical privacy} if for $\caD
    \sim \mu$, all $j \in [1:n]$ and $w,v \in W$ holds 
 \[   T_{M(F,\Djw),M(F,\Djv) }\gra T \ .\]
  }\end{definition}
  
 \noindent
To apply this notion certain operations on trade-off functions are needed:
    \begin{description}
        \item[Convex Conjugate:] $T^\star(y) \dea  \Sup{x} \ y \cdot x - T(x)$\\[-0.2ex]
        \item[Inverse:] $T^{-1}(\alpha) \dea \Inf{} \left\{ x \in  [0,1] \mid T(x) \leq \alpha \right\}$\\[-0.2ex]
        \item[Sample:] $T_p(x) \dea  p \; T(x) + (1-p) (1-x)$\\[-0.2ex]
        \item[$p$--sampling--operator:] $\Psi_p(T)  \dea  \Min{} \left(\left\{ T_p, T_p^{-1}
        \right\}^{\star}\right)^{\star}$
      \end{description} 
  
  
 \noindent  
The analysis that links $(\ee,\dd)$-DP  to trade-off functions 
can be adapted to statistical privacy.

  \begin{theorem}{\rm 
\label{theorem:SP_f-SP}  $ $ \\
For a distribution $\mu$ and a query $F$, a mechanism $M$ achieves
$T$-statistical privacy for the trade-off function defined by
\[ T(\alpha) \dea \Sup{\ee \geq 0} \;  \Max{}\{0, \ 1- \Phi_{\mu,F,M}(\ee) -
      e^\ee \alpha , \ e^{-\ee} (1 - \Phi_{\mu,F,M}(\ee) - \alpha)   \} . \]
Furthermore, if for a trade-off function $T$ one achieves  $T$-statistical
privacy then in the standard notion  one achieves $(\ee, 1+
T^\star(-e^\ee))$-statistical privacy.
  }\end{theorem}

  In case of subsampling the situation is more complex
  since the size of the database has a larger influence for statistical privacy.
  For differential privacy it is often assumed that
  the sensitivity is independent of the database size, which seems
  questionable, for example in case of property queries. 
  Under this assumption the effect of subsampling has been analysed  precisely. 
  If a mechanism $M$ achieves $(\ee,  \dd)$-DP 
  then $M \circ \SAMP_\T$ achieves $(\klog(1+\ll (e^\ee -1)), \ll{}\dd)$-DP
 for a sampling technique $T$ without
  replacement and sampling rate $\ll{}$ \cite{BBG18}.

For statistical privacy the uncertainty of an adversary is depends on the size of the
  database. Thus the techniques used in the the proofs for the $T$-DP
  subsampling theorem need to be adjusted to utilize the the parameters achieved
  when considering a smaller database.

\begin{theorem}\label{theorem:T_sp_without_replacement} \rm
Le $\omu$ be a distribution of entries, $\mu_n=\omu^n$ and $\mu_m=\omu^m$
be the corresponding product distributions for databases of size $n$, resp.~$m$ and
$\T_{n,m}$  sampling without replacement. 
If for  a trade-off function $T$ one can achieve  $T$-statistical privacy 
for $\mu_m$ and $F$   then $F \circ \SAMP_{\T_{n,m}}$ achieves 
$\Psi_\frac{m}{n}(T)$-statistical privacy for $\mu_n$.
\end{theorem}
\begin{proof}
    The proof uses the machinery of \cite{DRS22}, but instead of the
    sampling result for DP we use corollary
    \ref{corollary:subsampling_without_rep:iid} together with 
    Theorem~\ref{theorem:SP_f-SP}.
     $F$ is $(\ee, 1+ T_m^\star(-e^\ee))$--SP under $\mu_m$. 
     Thus $F \circ \SAMP_{\T_{n,m}}$
    is  $(\klog(1+ \frac{m}{n} \; (e^\ee -1)) ,\;  \frac{m}{n}   \;   (1+
    T^\star(-e^\ee)) )$--SP for $\mu_n$. 
\end{proof}



\section{Conclusion}
This paper has shows, that in the statistical adversarial model sampling can
further increase privacy significantly. To prove this result we have developed
methods to analyze the complex combination of sampling techniques and database
distributions. Compared to a worst case analysis in DP, in SP one has to
carefully keep track of all generated mixture distributions and the dependencies
in case of sampling with replacement. Our results are comparable to the results
for DP. 

We have applied our framework to property queries to get quantitative data. As a
next step one could consider other types of queries in the statistical setting.
The privacy utility tradeoff should be analyzed for sampling in more detail.
Furthermore, the effect of subsampling when composing queries needs to be
investigated.

\bibliographystyle{splncs04}
\bibliography{subsam}

\appendix

\section{Markov Kernels and Couplings}\label{section:markov_kernel_and_coupling}

Let us repeat some basic notions concerning Markov Kernels and Couplings.

\FNS{
Since many of the proof in this article rely on techniques and structures in the
field of stochastic processes we will give a short introduction of these. If we
consider any subset $A \subseteq \Rb$ of the real numbers we will alway consider
that as a measurable space with respect to the Lebesgue measure and its
Borel-$\sigma$-Algebra $\caB(S)$.
}

\begin{definition}{\rm \wfinw{Markov Kernel} \label{preliminaries:mk}\\
   Let $(X, \caX)$ and $(Y, \caY)$ be measurable spaces. 
   A function $\caK: \caY
   \times X \to [0,1]$ is called a \emph{Markov kernel} from $(X, \caX)$
   to $(Y, \caY)$ if the following properties are fulfilled:\\[-4ex]
   \begin{enumerate}
    \item For all fixed $\hat{Y} \in \caY$, the map $x \mapsto \caK(\hat{Y},x)$ is $\caX$-$\caB([0,1])$-measurable,
      where $\caB([0,1])$ denotes the Borel $\sigma$-algebra.
    \item For all fixed $x \in X$, the map $\hat{Y} \mapsto \caK(\hat{Y},x)$ is a probability measure on $(Y,\caY)$.
 \end{enumerate}

A function $\cak: Y \times X \to [0,\infty]$ is called
   the kernel density of $\caK$ with respect to a measure $\mu$ on $Y$ if
   for all $\hat{Y} \in \caY$ and $x \in X$ it holds
   \begin{equation*}
     \caK(\hat{Y},x) \gla \int_{\hat{Y}}   \cak(y,x) \; \mathbf{d}\mu(y).
   \end{equation*}
   }\end{definition}

    The Markov kernel of a measurable deterministic function $f$ is defined by its kernel
    density $\cak_f(x,y) = \hat{\delta}_{f(x)}(y)$ with respect to the Lebesgue
    measure. Thus, its Markov kernel equals
    \begin{equation*}
       \caK(\hat{Y},x) = \mathbf{1}_{\hat{Y}}(f(x)) = \mathbf{1}_{f^{-1}(\hat{Y})}(x) \ .
  \end{equation*}

    Furthermore,  for every function mapping elements of a set $X$ to random variables on a set $Y$
    there exists a Markov kernel describing its stochastic behavior,
    in particular for privacy techniques applied to a database query $F$.

For a  distribution $\mu$ on $(X,\caX)$,  a kernel $\caK: \caY \times X \to [0,1]$ from $(X, \caX)$
   to $(Y, \caY)$ acts on $\mu$ by 
   \begin{equation*}
    \mu \caK (\hat{Y}) := \int_X \caK(\hat{Y},x) \ \mathbf{d}\mu(x).
   \end{equation*}
Thus, $\mu \caK$ is a distribution on $(Y,\caY)$. 
For Markov kernels $\caK_1$ from $(X, \caX)$   to $(Y, \caY)$ and  $\caK_2$ from $(Y, \caY)$   to $(Z, \caZ)$ the composition
    $    \caK_1 \caK_2:  \caZ \times X \to [0,1]$
   is  a Markov kernel from $(X, \caX)$ to $(Z, \caZ)$.
    
 Couplings are a technique to compare two distributions $\mu$ and $\hat{\mu}$. 
These are 2-dimensional distributions whose marginal distributions are equal to $\mu$ and $\hat{\mu}$.
Of interest are couplings that maximizes the similarity of $\mu$ and $\hat{\mu}$,
 called a  maximal coupling.

   
   \begin{definition}{ \rm \bfinw{Maximal Couplings}\\
   Given two distributions $\mu$ and $\hat{\mu}$ whose total  variation equals $\ll$,
   a maximal coupling is a triple of distributions $\mu^0,\mu^1, \hat{\mu}^1$ such that
   \begin{eqnarray*}
      \mu &\gla&(1-\lambda) \;  \mu^0 + \lambda \; \mu^1 \\ 
      \hat{\mu} &\gla&(1-\lambda) \; \mu^0 + \lambda \; \hat{\mu}^1 \ .  
   \end{eqnarray*}
   }\end{definition}
   
   $\mu^0$ is the part where both distributions coincide and $ \mu^1, \hat{\mu}^1$ where
   they differ.

   A useful tool to analyze privacy is the \emph{advanced joint
   convexity} property of the $\alpha$--divergence (see \cite{BBG18}). 
   In our setting this corresponds to the  following: 
   
   \begin{lemma}{\rm \bfinw{Advanced Joint Convexity} \label{results:definition:AJK}\\
    For distributions $\mu,\hat{\mu}$ with total variation $\lambda$ and a  query $F$ 
       let $\mu^0,\mu^1,  \hat{\mu}^1$ be a maximal coupling.
       Then for  $\ee' = \log(1 + \lambda (e^\ee -1)) \approx \lambda \; \ee$   it holds
   \begin{eqnarray*}
      \DD_{F,\mu, \hat{\mu}}(\ee') &\kla& \lambda \left((1- e^{\ee' -\ee} )\; \DD_{F,\mu_1,
      \mu_0}(\ee) \; + \;  e^{\ee' -\ee} \; \DD_{F,\mu_1, \hat{\mu}_1}(\ee)  \right)
   \end{eqnarray*}
   }\end{lemma}

   \section{Class of $F$--samplable distributions}\label{appendix_samplable_bsp}
   There are many distributions commonly appearing in real world data making it
   necessary to try to analyze a wide array of different queries and database
   distributions. 
   \KURZ{}{\paragraph{Continuous Distributions}}
   The exponential family is a set of distributions where the density function
   $\mu$, which depends on some parameter vector $p \in \Pi$, can be expressed
   using functions $T: A \to \Rb^k, g: \Pi \to \Rb,\xi: \Pi \to \Rb^k$ and
   non-negative $h: A \to \Rb$ as
   \[
       \mu(x) := g(p) h(x) e^{\xi(p) T(x)^T}.
   \]
   \begin{theorem}
       Let $F$ be a query $F$ and $\caD$ a database. When for all $v, w \in W, j \in [1:n]$ and
       $\tau in \mathfrak{C}^n$ the distributions of
       $F(\SAMP_\tau \caD_{j,v})$ and $F(\SAMP_\tau \caD_{j,w})$ are part of the
       exponential family with parameters $p, p' \in \Pi$ such that
   \[
           \sum_{i=1}^k T_i(x) (\xi_i(p) - \xi_i(p'))
   \]
       is a polynomial of degree one or less. Then $F$ if $F$-samplable with respect
       to $\caD$.
   \end{theorem}
   \begin{proof}
       One needs to asses the set $S \sse A$ that maximizes the integral
       \[    \int_S  \mujv \caK_\tau \caK_{F}(x) - e^\ee \;
           \mujw \caK_{\tau} \caK_{F}(x)  \ \DX{x}  \]
       for all $\ee > 0$. Thus by evaluating the zero of $\mujv \caK_\tau
       \caK_{F}(x) - e^\ee \; \mujw \caK_{\tau} \caK_{F}(x)$ the structure of
       $S$ can be determined. Since the evaluated densities are in the
       exponential family this amounts to
       \begin{align*}
           0 &\gla g(p) h(x) e^{\xi(p) T(x)^T} - g(p') e^\ee h(x) e^{\xi(p') T(x)^T} \\
         \Leftrightarrow  0 &\gla \frac{g(p) e^{\sum_{i=0}^k \xi_i(p) T_i(x)} }{g(p') e^\ee e^{ \sum_{i=0}^k \xi_i(p') T_i(x)}} - 1 \\
         \Leftrightarrow  \log(1) &\gla \log \left( \frac{g(p)}{g(p') e^\ee } \prod_{i=0}^k \frac{e^{\xi_i(p) T_i(x)}}{e^{\xi_i(p') T_i(x)}} \right) \\
         \Leftrightarrow  0 &\gla \log \left( \frac{g(p)}{g(p') e^\ee } \right) + \sum_{i=0}^k T_i(x) (\xi_i(p) - \xi_i(p')). 
       \end{align*}
       When the last sum is a polynomial of degree zero the integrant is a constant
      $c$ allowing for $S = \emptyset$ depending of the sign of $c$. If the sum is
      a polynomial of degree one the last equality has one solution $x^\star$ for
      $x$ and thus $S = (- \infty, x^\star]$ or $S = [x^\star, \infty)$. 
   \end{proof}
   Since the way in which the final distribution of the query answer is
   generated this requirement needs to be evaluated for every case. 
   \KURZ{As a general Rule it holds often if the rate or variace parameters are
   fixed. This can be observed for the normal or Gamma distributions and holds
   aswell for discrete members of the exponential family like Binomial, Geometic or
   Poisson distributions. 
   }{\begin{example}
       Let $F$ be a query $F$ and $\caD$ a database. When for all $v, w \in W, j \in [1:n]$ and
       $\tau in \mathfrak{C}^n$ the distributions of
       $F(\SAMP_\tau \caD_{j,v})$ and $F(\SAMP_\tau \caD_{j,w})$ 
       \begin{itemize}
           \item have the same variance and are normal distributed.
           \item have the same rate parameter and are Gamma distributed.
       \end{itemize}
       In the first case is come from the fact that in the case of a normal
       distribution it holds $T(x) = (x^2, -x, 1)$ and $\xi((m,\sigma^2)) =
       1/\sigma^2 (1, 2m, m^2)$ therefore the $x^2$ summand cancels out since the
       variance is the same leaving a polynomial of first degree. For the Gamma
       distribution it holds $T(x) = -x$ and $\xi((\lambda, \alpha)) = \lambda$ for
       rate parameter $\ll$ and shape parameter $\alpha$. 
   \end{example}
   While it is not part of the exponential family if the query answer is Laplace
   distributed with fixed Variance as in the Normal distribution example this holds
   as well. This is shown in the proof for additive Laplace noise in \cite{BR25}.
   
   \paragraph{Discrete Distributions}
   As seen in the case for continuous distributions the property is fulfilled for
   distributions where the probability mass is centered at one point an then falls
   exponentially fast. This is mimicked by several discrete distributions like the
   Binomial, Geometric and Poisson distributions which all are part of the discrete
   exponential family. To see that query answers distributed like this are
   generated by a $F$-samplable query and fitting distribution (with the requirement
   for the Binomial distribution that the $n$ parameter is fixed) one need to
   observe, that they all depend on a parameter $p > 0$ which is the base of an
   exponentiation and the possible elements of the sample space are the exponent.
   If we then compare two such distributions with different parameters even when
   one is scaled by a constant (here $e^\ee$) since both decrease exponentially
   fast one will always end up lower then the other at one fixed point. }
   
   \paragraph{Sum Queries}
   \KURZ{Sum queries, $g_{\alpha}: \Rb^n \to \Rb$ with $g_{\alpha}(x) = \sum_{i=0}^n \alpha_i x_i$ for
   some fixed $\alpha \in \Rb^n$, are a common type of query considered in
   differential privacy. 
   \begin{theorem}
       Let $\caD$ be a database and $\alpha \in \Rb^N$ such that for all $i \in
       [1:n]$ the random variables $\alpha_i \caD_i$ are identically distributed
       and they are all either \emph{Binomial, Poisson, Normal} or  \emph{Cauchy}
       distributed then $g_\alpha$ if $F$--samplable with respect to $\caD$
   \end{theorem}
   \begin{proof}
       All distributions from the proposition are closed with respect to summation.
       And the first three are members of the exponential family where the
       problematic parameters are fixed. For Cauchy distributed answers the zero of
       the difference of the densities to $Cau(a, \gamma)$ and $Cau(a', \gamma)$
       needs to be evaluated.
       \begin{align*}
           0 &\gla \frac{1}{\pi} \left( \frac{\gamma}{(x-a')^2 + (\gamma)^2} - \frac{ e^\ee \gamma}{(x-a)^2 + (\gamma)^2} \right) \\
         \Leftrightarrow  1 &\gla \frac{(x-a)^2 + \gamma^2}{(x-a')^2 + \gamma^2}
       \end{align*}
       The second equality holds iff \[ (x-a)^2 = (x-a')^2 \] solving for $x^\star
       = ((a')^2 - a^2)/2(a'-a)$ and thus $S = (- \infty, x^\star]$ or $S = [x^\star, \infty)$. 
   \end{proof}
   }{A type of query often considered in differential Privacy are sum queries 
   $g_{\alpha}: \Rb^n \to \Rb$ with $g_{\alpha}(x) = \sum_{i=0}^n \alpha_i x_i$ for
   some fixed $\alpha \in \Rb^n$. One recognizes that there are choices for $\alpha$
   where these queries might fulfill neither DP or SP. But many queries like the
   mean fall into this class of queries. 
   While this motivates why one should analyze this type of query it has
   additionally useful mathematical properties since the sum of two distributions
   is evaluated by their convolution which are well studied. 
   \begin{theorem}
       Let $\caD$ be a database and $\alpha \in \Rb^N$ such that for all $i \in
       [1:n]$ the random variables $\alpha_i \caD_i$ are identically distributed
       and they are all either \emph{Binomial, Poisson, Normal} or  \emph{Cauchy}
       distributed then $g_\alpha$ if $F$--samplable with respect to $\caD$
   \end{theorem}
   \begin{proof}
       Since for all of the distributions in the list we know, that sums of random
       variables distributed like that are again of the same distribution we only
       need to analyze the answer distribution. For the first three points of the
       list this was already discussed. Thus we only need to determine the set $S
       \sse A$ that maximizes the integral
       \[    \int_S  \mujv \caK_\tau \caK_{g_{\SAMP_\tau(\alpha)}}(x) - e^\ee \;
           \mujw \caK_{\tau} \caK_{g_{\SAMP_\tau(\alpha)}}(x)  \ \DX{x}  \] 
           for all $\ee > 0$ for Cauchy distributed answers. We again do this by
           evaluating the zero assuming the answers are $Cau(a, \gamma)$ and
           $Cau(a', \gamma)$ distributed (since the two answers are of the same
           distribution and differ only by their mean we consider the case where
           their variance is equal). 
       \begin{align*}
           0 &\gla \frac{1}{\pi} \left( \frac{\gamma}{(x-a')^2 + (\gamma)^2} - \frac{ e^\ee \gamma}{(x-a)^2 + (\gamma)^2} \right) \\
         \Leftrightarrow  1 &\gla \frac{(x-a)^2 + \gamma^2}{(x-a')^2 + \gamma^2}
       \end{align*}
       The second equality holds iff \[ (x-a)^2 = (x-a')^2 \] solving for $x^\star
       = ((a')^2 - a^2)/2(a'-a)$ and thus $S = (- \infty, x^\star]$ or $S = [x^\star, \infty)$. 
   \end{proof}
   
   In addition to these cases where the distribution of the sum is known the
   central limit theorem implies that for big enough databases the answer
   distribution converges weakly against a normal distribution and thus the
   probability mass which leads to a structure of the set $S$ such that it does not
   fulfill the requirements stated in $F$--sampeble should converge to $0$. This
   behavior and boarder the influence this convergence behavior has of the
   different divergences used in differential privacy should be focus of further
   research. 
   }
   
   \KURZ{\section{Poisson Sampling in Differential
   Privacy}\label{appendix:poission_DP_vs_SP}
   Sampling for differential privacy is well studied (see among others
   \cite{LQS12,BBG18,BBG20}), the classic results state, that when sampling with
   rate $\ll$ without replacement a mechanisms fulfilling $(\ee,\dd)$--DP becomes
   $(\log(1+\ll (e^\ee -1)),\ll \;\dd)$--DP. In difference to the bounds for
   Poisson sampling in \cite{BBG18} we propose a new bound for the setting of 
   
   \begin{theorem}
       Let $F$ be a query and $M$ some privacy technique achieving DP for a privacy
       curve $\dd(\ee)$. Then the privacy technique $M(F, \SAMP_\T(\cdot))$ where $\T$ is a Poisson
       sampling technique with sampling rate $\ll$ achieves DP for the privacy
       curve
       \[ \dd^\star(\ee) = \sum^n_{m=0} ll^m (1-\ll)^{n-m} \binom{n}{m} m/n
       \dd(\log(1+ n/m (e^\ee
       -1))). \]
   \end{theorem}
   
   \begin{proof}
       We follow our proof in \ref{theorem:subsampling:poisson} for our second
       bound. Here $\caK_{M,F}$ defines the Markov kernel corresponding to the
       application of $M(F, \cdot)$. Thus for any $D,D' \in W^\star$ with $D
       \approx_c D'$ where $\mu_D$ and $\mu_{D'}$ define their point distributions
       we can represent their privacy curve $d(\ee) = \DD_{(F,M),\mu_D \caK_\T, \mu_{D'} \caK_\T}(\ee)$.
       This allows for the same steps to split along the database size as in the SP proof leaving us with 
       \[ d(\ee) \leq \sum^n_{m=0} ll^m (1-\ll)^{n-m} \binom{n}{m} \underbrace{ \Max{S \subseteq A} \int_S \mu_D \caK_{\T^m} \caK_{M,F}(x)  - e^\ee
       \mu_{D'} \caK_{\T^m} \caK_{M,F}(x) \DX{x}}_{(A)}. \]
       As for SP one recognizes that $(A)$ corresponds to the sampling without
       replacement technique applied to $M$ and $F$ with a sample size of $m$.
       Bounding $(A)$ by choosing the maximum over all neighboring databases and
       applying the amplification bound for sampling without replacement on
       databases that differ in one entry stated in \cite{BBG20} we arrive at the
       Bound stated in the theorem. 
   \end{proof}
   
   }{\section{Comparison DP and SP for Poisson
   sampling}\label{appendix:poission_DP_vs_SP}
   In \cite{LQS12} a tight result for the effect of Poisson sampling in
   differential privacy was given and later reformulated in \cite{BBG20}. This
   considers for some database space $W^\star$ the neighborhood relationship
   $\approx_r$ for $D, D' \in W^\star$ as $D \approx_r D'$ iff $\vert D \vert - 1 =
   \vert D' \vert$ and $D' \subset D$ or $\vert D' \vert - 1 =
   \vert D \vert $ and $D \subset D'$, meaning two databases are neighboring iff
   one can be created by removing one entry from the other. This is equivalent to
   the common notion where two database are considered neighboring iff they only
   differ in one entry.
   
   \begin{theorem}
       Let $F$ be a query and $M$ some privacy technique achieving $(\ee,\dd)$--DP.
       Then the privacy technique $M(F, \SAMP_\T(\cdot))$ where $\T$ is a Poisson
       sampling technique with sampling rate $\ll$ achieves $(\log(1+\ll (e^\ee
       -1)),\ll \;\dd)$--DP.
   \end{theorem}
   Taking the assumption in DP into account, that the size of $\dd$ is independent
   of the size of the examined database. 
   
   As our Simulations showed there may be instances in which our second bound is
   stronger using the similarities of SP and DP the following Theorem shows a new
   bound for Poisson sampling under the classic neighborhood definition where two
   databases $D, D' \in W^n$ are $D \approx_c D'$ iff $\exists! j \in [1:n]$
   s.t. $D_i \neq D_i'$.
   
   \begin{theorem}
       Let $F$ be a query and $M$ some privacy technique achieving DP for a privacy
       curve $\dd(\ee)$. Then the privacy technique $M(F, \SAMP_\T(\cdot))$ where $\T$ is a Poisson
       sampling technique with sampling rate $\ll$ achieves DP for the privacy
       curve
       \[ \dd^\star(\ee) = \sum^n_{m=0} ll^m (1-\ll)^{n-m} \binom{n}{m} m/n
       \dd(\log(1+ n/m (e^\ee
       -1))). \]
   \end{theorem}
   \KURZ{\begin{proof}
       We follow our proof in \ref{theorem:subsampling:poisson} for our second
       bound. Here $\caK_{M,F}$ defines the Markov kernel corresponding to the
       application of $M(F, \cdot)$. Thus for any $D,D' \in W^\star$ with $D
       \approx_c D'$ where $\mu_D$ and $\mu_{D'}$ define their point distributions
       we can represent their privacy curve 
       \[ d(\ee) = \DD_{}
        \Max{S \subseteq A} \int_S \mu_D \caK_\T \caK_{M,F}(x)  - e^\ee
       \mu_{D'} \caK_\T \caK_{M,F}(x) \DX{x} \]
       This allows for the same steps to split along the database size as in the SP proof leaving us with 
       \[ d(\ee) \leq \sum^n_{m=0} ll^m (1-\ll)^{n-m} \binom{n}{m} \underbrace{ \Max{S \subseteq A} \int_S \mu_D \caK_{\T^m} \caK_{M,F}(x)  - e^\ee
       \mu_{D'} \caK_{\T^m} \caK_{M,F}(x) \DX{x}}_{(A)}. \]
       As for SP one recognizes that $(A)$ corresponds to the sampling without
       replacement technique applied to $M$ and $F$ with a sample size of $m$.
       Bounding $(A)$ by choosing the maximum over all neighboring databases and
       applying the amplification bound for sampling without replacement on
       databases that differ in one entry stated in \cite{BBG20} we arrive at the
       Bound stated in the theorem. 
   \end{proof}}{\begin{proof}
       We follow our proof in \ref{theorem:subsampling:poisson} for our second
       bound. Here $\caK_{M,F}$ defines the Markov kernel corresponding to the
       application of $M(F, \cdot)$. Thus for any $D,D' \in W^\star$ with $D
       \approx_c D'$ where $\mu_D$ and $\mu_{D'}$ define their point distributions
       we can represent their privacy curve 
       \[ d(\ee) = \Max{S \subseteq A} \int_S \mu_D \caK_\T \caK_{M,F}(x)  - e^\ee
       \mu_{D'} \caK_\T \caK_{M,F}(x) \DX{x} \]
       This allows for the same steps to split along the database size as in the SP proof leaving us with 
       \[ d(\ee) \leq \sum^n_{m=0} ll^m (1-\ll)^{n-m} \binom{n}{m} \underbrace{ \Max{S \subseteq A} \int_S \mu_D \caK_{\T^m} \caK_{M,F}(x)  - e^\ee
       \mu_{D'} \caK_{\T^m} \caK_{M,F}(x) \DX{x}}_{(A)}. \]
       As for SP one recognizes that $(A)$ corresponds to the sampling without
       replacement technique applied to $M$ and $F$ with a sample size of $m$.
       Bounding $(A)$ by choosing the maximum over all neighboring databases and
       applying the amplification bound for sampling without replacement on
       databases that differ in one entry stated in \cite{BBG20} we arrive at the
       Bound stated in the theorem. 
   \end{proof}}

   }
   
   \end{document}